\documentclass[copyright,creativecommons]{eptcs}

\usepackage{algo}
\usepackage{amssymb}
\usepackage{amsmath}
\usepackage{amsthm}
\usepackage{breakurl,cite}
\usepackage[mathscr]{euscript}
\usepackage{multirow}
\usepackage{graphicx}
\usepackage{graphs}
\usepackage{mathtools}
\usepackage{times}
\usepackage[svgnames,dvips]{xcolor}
\usepackage[T1]{fontenc}

\providecommand{\event}{FMSPLE 2015}
\providecommand{\firstpage}{1}
\providecommand{\sortnoop}[1]{}

\mathcode`;="203B       
\mathcode`?="003F       
\mathcode`|="026A       

\mathchardef\ls="213C   
\mathchardef\gr="213E   
\mathchardef\uparrow="0222      
\mathchardef\downarrow="0223    

\newcommand{\blankline}{\vspace*{0.5\baselineskip}}
\newcommand{\quarterlineup}{\vspace*{-0.25\baselineskip}}
\newcommand{\nop}[1]{}

\newcommand{\FExp}{\mathbb{B}( \mkern1mu \calF \mkern1mu )}
\newcommand{\FExphat}{\widehat{\FExp\rule{0pt}{9pt}}}
\renewcommand{\FExphat}{\widehat{\mathbb{B}}( \mkern1mu \calF \mkern1mu )}
\newcommand{\TRUE}{\textit{true}}
\newcommand{\TRUEhat}{\widehat{\TRUE}}
\newcommand{\FALSE}{\textit{false}}

\newcommand{\Trans}{\, \mathrel{\xRightarrow{\ }} \,}

\newcommand{\assign}{\mathrel{{:}{=}}}
\newcommand{\bbisim}{\simeq_{b}}
\newcommand{\bbisimP}{\simeq^{b}_{P}}
\renewcommand{\bbisimP}{\simeq_{P}}
\newcommand{\bfbisim}{\simeq_{\mathit{bf}}}

\newcommand{\bftwo}{\textbf{2}}

\newcommand{\calAf}{\calA_{\mkern-2mu f}}
\newcommand{\calAtau}{\calA_\tau}
\newcommand{\calBmin}{\calB_{\mathit{min}}}
\newcommand{\calSmin}{\calS_{\mathit{min}}}
\newcommand{\fTrans}[1]{\, \xRightarrow{#1} \,}

\newcommand{\cbfbisim}{\simeq_{\mathit{cbf}}}

\newcommand{\ftrans}[2]{\, \xrightarrow{{#1}|\mkern1mu{#2}} \,}
\newcommand{\ftranss}[1]{\, \xrightarrow{{#1}}\hspace{-14pt}\longrightarrow \,}
\newcommand{\lc}{\mathopen{\lbrace \;}}
\renewcommand{\ftranss}[1]{\fTrans{#1}}
\newcommand{\notsimP}{\not\sim_\calP}
\newcommand{\ntrans}[2]{\, \stackrel{#1}{\nrightarrow}_{#2} \,}
\newcommand{\nTrans}{\, \nRightarrow \,}

\newcommand{\rc}{\mathclose{\; \rbrace}}
\newcommand{\shat}{\hat{s}}
\newcommand{\simP}{\sim_\calP}

\newcommand{\singleton}[1]{\lbrace {#1} \rbrace}
\newcommand{\sinit}{s^{\mbox{\tiny $\circ$}}}
\renewcommand{\sinit}{s_{\mkern-1mu {\ast}}}
\newcommand{\that}{\hat{t}}

\newcommand{\trans}[2]{\, \xrightarrow{#1}_{#2} \,}
\newcommand{\transP}[1]{\trans{#1}{P}}
\newcommand{\varphihat}{\hat{\varphi}}

\newcommand{\calA}{\mathscr{A}}
\newcommand{\calB}{\mathscr{B}}
\newcommand{\calC}{\mathscr{C}}
\newcommand{\calF}{\mathscr{F}}
\newcommand{\calG}{\mathscr{G}}

\newcommand{\calP}{\mathscr{P}}
\newcommand{\calS}{\mathscr{S}}
\newcommand{\calT}{\mathscr{T}}
\newcommand{\calU}{\mathscr{U}}

\newtheorem{theorem}{Theorem}
\newtheorem{lemma}[theorem]{Lemma}
\newtheorem{definition}[theorem]{Definition}
\newtheorem{corollary}[theorem]{Corollary}

\title{Coherent branching feature bisimulation}

\author{%
  Tessa Belder
  \institute{TU/e, Eindhoven\\ The Netherlands}
  \and
  Maurice H. ter Beek
  \institute{ISTI--CNR, Pisa\\ Italy}
  \and
  Erik P. de Vink\thanks{Corresponding author, email~\url{evink@win.tue.nl}.}
  \institute{TU/e, Eindhoven \& CWI, Amsterdam\\ The Netherlands}
}

\def\titlerunning{Coherent branching feature bisimulation}
\def\authorrunning{T. Belder, M.H. ter Beek \& E.P. de Vink}

\begin{document}

\maketitle

\begin{abstract} 
  Progress in the behavioral analysis of software product lines at the
  family level benefits from further development of the underlying
  semantical theory. Here, we propose a behavioral equivalence for
  feature transition systems (FTS) generalizing branching bisimulation
  for labeled transition systems (LTS)\@. We prove that branching
  feature bisimulation for an FTS of a family of products coincides
  with branching bisimulation for the LTS projection of each the
  individual products. For a restricted notion of coherent branching
  feature bisimulation we furthermore present a minimization algorithm
  and show its correctness. Although the minimization problem for
  coherent branching feature bisimulation is shown to be intractable,
  application of the algorithm in the setting of a small case study
  results in a significant speed-up of model checking of behavioral
  properties.
\end{abstract}


\section{Introduction}
\label{sec-intro}

Notions of behavioral equivalence, like bisimulation, play an
important role in the analysis of large systems in general and thus of
(software) product lines in particular. Abstractions based on
behavioral equivalences compress, via abstraction operations and
minimization algorithms, a model's state space prior to
verification. Subsequently, verification can be done in less time,
using less memory.

Compared to single system verification, SPLE adds variability as yet
another dimension to the complexity of behavioral analysis. In
general, the number of possible products of a product line is
exponential in the number of features. This calls for dedicated
modeling and analysis techniques that allow to specify and reason
about an entire product line at once. In this paper we consider the
model of feature transition systems~\cite{CHSLR10:icse,CCHSLR13},
which facilitates efficient family-based verification. Dedicated
techniques generally use variability knowledge about valid feature
configurations to deduce results for products from a family model, as
opposed to enumerative product-based verification, in which every
product is examined individually. For example, in~\cite{CCPSHL12}
behavioral pre-orders of FTS are given with respect to specific
products to define abstractions based on simulation quotients that
preserve LTL properties. We refer to~\cite{TAKSS14} for an overview of
verification strategies in SPLE and the trade-off of product-based
vs.\ family-based analysis.

In~\cite{BV14:formalise,BV14b} we applied tailored property preserving
reductions to a product line modeled with mCRL2~\cite{CGK+13:tacas}
and we verified by means of model checking a number of behavioral
properties of the product line. The mCRL2 toolset provides specific
support for reduction modulo branching
bisimulation~\cite{GW96:jacm}. This led us to investigate a
feature-oriented notion of branching bisimulation inspired by the
research reported in~\cite{CCPSHL12} (which focuses on a notion of
simulation). In this paper, we propose a definition of what is coined
branching feature bisimulation, extending the definition
in~\cite{GW96:jacm}, and we seek to adapt the efficient algorithm
of~\cite{GV90:icalp} to compute, given an FTS, a minimal FTS that is
branching feature bisimilar.

In our pursuit to transfer the results of~\cite{CCPSHL12} to the case
of branching bisimulation, a number of issues arises due to the
presence of feature expressions, though. One such issue for FTS is
that minimization in the number of states is not the same as
minimization in the number of transitions, a situation that does not
occur with LTS\@. Our effort here is to reduce in the number of
states. In order to make our minimization algorithm work, we
  restrict to so-called coherent rather than arbitrary branching
  feature bisimulation relations.  We will prove that our algorithm
reduces an FTS~$\calS$ to a minimal FTS~$\calSmin$ for which there
exists a coherent branching feature bisimulation relation for $\calS$
and~$\calSmin$. Moreover, no smaller FTS~$\calS'$ exists
  which is also coherent branching feature bisimilar to~$\calS$.
However, as we will argue by a reduction of graph coloring, the
minimization problem is NP-complete for coherent branching feature
bisimulations (and we suspect this is the case for branching
  feature bisimulation as well). Still, as an evaluation of the approach
for a relatively small toy example illustrates, overall a substantial
reduction in computation time is achieved for bisimulation-enhanced
family-based analysis as compared to enumerative product-based
analysis. In particular, for properties involving a limited number of
features, verification time using the family FTS is only a third to a
quarter of the time needed to verify all product LTS\@.

Behavioral equivalences also form the basis of conformance 
notions as used for model elaboration by iterative refinement of 
partial behavioral models. In SPLE, this allows to relate fully 
configured product behavior to family models with optional 
behavior reflecting product variability. Examples are approaches 
based on process algebra~\cite{Tri14} and on modal transition 
systems (MTS)~\cite{FUB06,ABFG11b,ABFG12}. In~\cite{Tri14}, 
a so-called variant process algebra is introduced, which allows to 
model family behavior that subsumes the behavior of all possible 
product variants. Special-purpose bisimulation relations then 
allow to compare variants among each other and against the 
family.
In SPLE, MTS are one of the models used to specify family behavior
encompassing all possible product behavior, represented by those LTS
that are implementations of the MTS (obtained by refinement of
admissible behavior). In~\cite{FUB06}, weak and strong refinement for
MTS as defined in~\cite{LT88} (based on weak and strong bisimulation)
are shown to be inadequate for applications in SPLE (mainly due to the
lack of support for unobservable actions and for preserving branching
behavior, respectively) and a novel notion of refinement is introduced
preserving the branching structure. It moreover preserves properties
expressed in $3$-valued weak $\mu$-calculus. However, its definition
is not operational and algorithms for conformance checking conformance
are thus infeasable.

The paper outline is as follows.  Building on definitions and an
algorithm for branching bisimulation of LTS reviewed in
Section~\ref{sec-branching-bisim}, we introduce in
Section~\ref{sec-branching-feature-bisim} the notion of branching feature
bisimulation and show its soundness for branching bisimulation with
respect to all products. The algorithm for minimizing modulo coherent
branching feature bisimulation is given in
Section~\ref{sec-algorithm}, which also provides an NP-completeness
proof for the minimization problem. A validation of the approach,
based on a toy example of a product line of coffee/soup vending machines
is reported in Section~\ref{sec-experiments}. Finally,
Section~\ref{sec-conclusion} briefly wraps up with concluding remarks
and future work.


\section{Branching bisimulation for labeled transition systems}
\label{sec-branching-bisim}

Strong bisimulation is a cornerstone of the theory of LTS~\cite{Mil89}, but is
often too fine a behavioral equivalence for verification
purposes. Application of its minimization algorithm typically reduces
the system under verification only in a limited way.  Having this in
mind, various weaker notions have been studied in the
literature~\cite{Gla90:concur,Gla93:concur}. In the context of model
checking, branching bisimulation as proposed for LTS by Van Glabbeek
\& Weijland enjoys a number of appealing
properties~\cite{GW89:ifip}. We recall and illustrate its definition,
and discuss the outline of a minimization algorithm that returns the
smallest LTS that is branching bisimilar to a given one. To this end,
we fix an alphabet of actions $\calA$, distinguish a symbol
$\tau\notin\calA$, referred to as the silent action, and let $\calAtau
= \calA \cup\singleton{\tau}$.

\blankline

\begin{definition}
  \label{def-branching-bisim}
  A labeled transition system is a triple $\calS = ( S ,\,
  {\rightarrow} ,\, \sinit )$ with set of states~$S$, transition
  relation ${\rightarrow} \subseteq {S \times \calAtau \times S}$, and
  initial state~$\sinit \in S$.
  \begin{itemize}
  \item [(a)] For $s,s' \in S$, we write
    $s \Trans s'$ if\, $\exists \mkern1mu n\,\exists \mkern1mu s_0
    \cdots s_n \colon s_0 = s \land \bigl( \forall i ,\, 1 \leqslant i
    \leqslant n \colon s_{i-1} \trans{\,\tau\,}{} s_i \bigr) \land s_n = s'$.
  \item [(b)] A symmetric relation~$R \subseteq {S \times S}$ is
    called a branching bisimulation relation if $\forall s,s',t \in
    S$, $\alpha \in \calAtau$ such that $R(s,t)$ and $s
    \trans{\alpha}{} s'$, it holds that $R(s,\that \mkern1mu )$,
    $R(s',t')$ and $t \Trans \that \trans{(\alpha)}{} t'$ for some
    $\that, t' \in S$.
  \item [(c)] Two states $s,t$ of~$\calS$ are called branching
    bisimilar if $R(s,t)$ for some branching bisimulation
    relation~$R$. Notation~$s \bbisim t$.
  \end{itemize}
\end{definition}

\blankline

\noindent
Note the notation $\that \trans{(\alpha)}{} t'$ used in part~(b) of
this definition. Following~\cite{GW96:jacm}, we have
$\that\trans{(\alpha)}{} t'$ if either $\that \trans{\alpha}{} t'$ or
$\alpha = \tau$ and~$\that = t'$, an elegant trick to allow the
transition $s \trans{\tau}{} s'$ to be matched by $t = \that = t'$,
i.e.\ by no transition for~$t$ in case~$R(s',t)$.

In Figure\,\ref{fig-examples-branching-bisim} at the left-hand side,
$s_0$ and $t_0$ are not branching bisimilar:\,Clearly state $s_1$ is
not branching bisimilar to state~$t_0$ since $s_1$ has no
$b$-transition. But then, the transition $t_0 \trans{a}{} t_2$ cannot
be matched by the transition sequence $s_0 \Trans s_1 \trans{a}{} s_2$
because the intermediate state~$s_1$ cannot be related to state~$t_0$,
as specifically required by the definition. However, for $u_0$
and~$v_0$ at the right-hand side, the transition $v_0 \trans{a}{}
v_1$ can be matched by $u_0 \Trans u_1 \trans{a}{} u_4$, since in this
case $v_0$ and~$u_1$ are branching bisimilar. It is noted that
  $u_0$ and~$v_0$, but also $s_0$ and~$t_0$, are weakly bisimilar in
  the sense of Milner~\cite{Mil89}.

\begin{figure}[h]
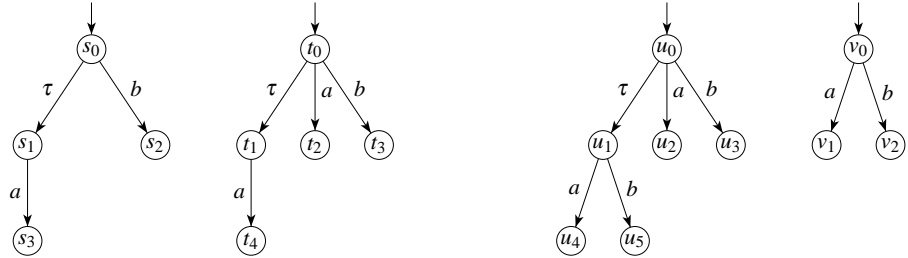

  \centering
  \small
  \hspace*{0.5cm}
  \scalebox{0.85}{%
  \begin{digraph}
    (200,40)(0,12.5)
    \graphset{Nadjust=n,Nw=4.5,Nh=4.5}
    \graphset{AHLength=2,AHangle=20}
    \node(s0)(25,45){$s_0$}
    \node(s1)(15,30){$s_1$}
    \node(s2)(35,30){$s_2$}
    \node(s3)(15,15){$s_3$}
    \imark[iangle=90](s0)
    \edge[ELside=r](s0,s1){$\tau$}
    \edge(s0,s2){$b$}
    \edge[ELside=r](s1,s3){$a$}
    \put(35,0){%
    \node(t0)(25,45){$t_0$}
    \node(t1)(15,30){$t_1$}
    \node(t2)(25,30){$t_2$}
    \node(t3)(35,30){$t_3$}
    \node(t4)(15,15){$t_4$}
    \imark[iangle=90](t0)
    \edge[ELside=r](t0,t1){$\tau$}
    \edge[ELpos=45](t0,t2){$a$}
    \edge(t0,t3){$b$}
    \edge[ELside=r](t1,t4){$a$}
    } 
    \put(90,0){%
    \node(u0)(25,45){$u_0$}
    \node(u1)(15,30){$u_1$}
    \node(u2)(25,30){$u_2$}
    \node(u3)(35,30){$u_3$}
    \node(u4)(10,15){$u_4$}
    \node(u5)(20,15){$u_5$}
    \imark[iangle=90](u0)
    \edge[ELside=r](u0,u1){$\tau$}
    \edge[ELpos=45](u0,u2){$a$}
    \edge(u0,u3){$b$}
    \edge[ELside=r](u1,u4){$a$}
    \edge(u1,u5){$b$}
    } 
    \put(120,0){%
    \node(v0)(25,45){$v_0$}
    \node(v1)(20,30){$v_1$}
    \node(v2)(30,30){$v_2$}
    \imark[iangle=90](v0)
    \edge[ELside=r](v0,v1){$a$}
    \edge(v0,v2){$b$}
    } 
  \end{digraph}
  } 
  \caption{\label{fig-examples-branching-bisim}Two non-branching bisimilar states and two branching bisimilar states}
\end{figure}

\noindent
An efficient minimization algorithm for branching bisimulation is due
to Groote \& Vaandrager~\cite{GV90:icalp}, based on the partition
refinement algorithm of Paige \& Tarjan~\cite{PT87:siam}.  It involves
the notions of a partition of the set of states, and of a
splitter: Consider a finite LTS $\calS = ( S , {\trans{}{}} , \sinit
)$ over the action set~$\calAtau$.
\begin{itemize}
\item A partition of~$\calS$ is a collection $\calB = \lc B_i \mid
  i \in I \rc$ of subsets of~$S$ that disjointly covers~$S$, i.e.\,$\bigcup_{i \in I} \: B_i =\linebreak S$, and $B_i \cap B_j = \varnothing$ if $i
  \neq j$, for all $i,j \in I$. The elements of a partition are
  referred to as blocks.
\item For a partition~$\calB$, blocks~$B,B' \in \calB$, and
  $\alpha \in \calAtau$ we let $\mathit{pos}_\alpha(B,B')$ $=$ $\lc s
  \in B \mid \exists \mkern1mu \shat \in B \, \exists \mkern1mu s' \in
  B' \colon s \, \Trans \shat \trans{\alpha}{} s' \rc$, and
  $\mathit{neg}_\alpha(B,B')$ $=$ $\lc s \in B \mid \forall \shat \in
  B \mkern1mu \, \forall s' \in B' \mkern1mu \colon (s \nTrans
  \shat)\,\lor\,(\shat \ntrans{\alpha}{} s') \rc$.
\item For blocks~$B,B'$ of a partition~$\calB$, the block~$B'$ is
  called a splitter of~$B$ for an action $\alpha \in \calAtau$ if both
  $\mathit{pos}_\alpha(B,B') \neq \varnothing$ and
  $\mathit{neg}_\alpha(B,B') \neq \varnothing$.
\end{itemize}
A simplified version of the algorithm of~\cite{GV90:icalp} for
minimization modulo branching bisimulation starts with the trivial
partition~$\calB = \singleton{S}$ and iterates
\begin{center}
  \textbf{while} 
  splitter $B'$ of block $B \in \calB$ for~$\alpha \in \calAtau$ exists 
  \textbf{do}
  $\calB \assign ( \calB {\setminus}
  \singleton{B} ) \, \cup \, \lbrace \mkern1mu \mathit{pos}_\alpha(B,B') ,\,
  \mathit{neg}_\alpha(B,B') \mkern1mu \rbrace$ 
  \textbf{end}
\end{center}
Thus, starting from the trivial partition $\singleton{S}$, having the
complete set of states~$S$ as a single block, we keep refining the
partition based on a splitter. Clearly, the algorithm terminates for a
finite LTS in at most ${|}S{|}$ many steps. We refer
to~\cite{GV90:icalp} for a proof of the following result.


\begin{theorem}
  \label{th-correctness-branching-bisim}
  Assume $\calBmin$ is the partition obtained upon termination
  after applying the algorithm to the LTS $\calS = ( S ,\,
  {\trans{}{}} ,\, \sinit )$. Define the LTS $\calS_{\mathit{min}} = (
  \calBmin ,\, {\rightarrow_{\mathit{min}}} ,\, B_{\mkern-1mu \ast} )$
  by letting $B \xrightarrow{\alpha}_\mathit{min} B'$ if there exist
  $s \in B$, $s' \in B'$ such that $s \trans{\alpha}{} s'$ for $B,B'
  \in \calB$, $\alpha \in \calAtau$ with $B \neq B'$ or $\alpha \neq
  \tau$, and by choosing $B_{\mkern-1mu \ast}$ such that $\sinit \in
  B_{\mkern-1mu \ast} $. Then $\calS_{\mathit{min}}$ is the smallest
  LTS that is branching bisimilar to~$\calS$.  \qed
\end{theorem}

\blankline

\noindent
In the simplified algorithm sketched above, major part of the
computation is spent on unfolding of the relation~$\Trans$. The
algorithm of~\cite{GV90:icalp} reduces this by eliminating
$\tau$-cycles and by keeping track, per block, of so-called bottom
states. The complexity of the Groote \& Vaandrager algorithm is $O(m
\log m + m \cdot n)$, with~$n$ the number of states and~$m$ the number
of transitions. Typically, for an LTS~$m \ll n^2$. It is known that
branching bisimulation preserves the fragment of the modal
$\mu$-calculus consisting of CTL${}^\ast$ minus the next
operator~\cite{DV95:jacm}. Therefore, exploiting this fact in
practical situations, significant reduction of the state space and
corresponding speed-up of subsequent verification can be obtained by
applying hiding of action followed by the minimization algorithm for
branching bisimulation. 

In the sequel of this paper, we seek to apply the idea of
  branching bisimulation (i.e.\ allowing silent moves through bisimulation
  equivalent states but through no other) and its minimization
  techniques to the setting of FTS, where not only actions but also
  feature expressions decorate the transitions.


\section{Branching bisimulation for feature transition systems}
\label{sec-branching-feature-bisim}

We fix a finite non-empty set~$\calF$ of features, a subset
$\calP\subseteq \bftwo^\calF$ of products, and again a set~$\calAtau$
including the silent action~$\tau$. We let~$\FExp$ denote the set of
boolean expressions over~$\calF$. We refer to elements of~$\FExp$ as
feature expressions. For a product~$P\in \calP$, we use~$\chi(P)$ to
denote its characteristic formula. The notion of a feature transition
system (FTS) was proposed in~\cite{CHSLR10:icse}.

\blankline

\begin{definition}
  A feature transition system (FTS) $\calS$ is a triple $\calS = ( S
  ,\, \theta ,\, \sinit )$, with $S$ the set of states, $\theta : S
  \times \calAtau \times S \to \FExp$ the transition constraint
  function, and~$\sinit \in S$ the initial state.
\end{definition}

\blankline

\noindent
For states~$s, s' \in S$, an action~$\alpha \in \calAtau$ and a
satisfiable feature expression $\psi \in \FExp$, we write $s
\ftrans{\alpha}{\psi} s'$ if $\theta(s,\alpha,s') = \psi$. We say that
a product~$P \in \calP$ satisfies a feature expression~$\varphi \in
\FExp$ if $\varphi$~is valid when the boolean variables corresponding
to the features of~$P$ are assigned the value~$\TRUE$ and those not
in~$P$ the value~$\FALSE$, denoted by $P \models \varphi$. The
equivalence relation~$\simP$ on~$\FExp$ is given by $\varphi \simP
\psi$ iff $\forall P \in \calP$: $P \models \varphi \Leftrightarrow P
\models \psi$. We let $\FExphat = \FExp / {\simP}$. For an FTS $\calS
= ( S ,\, \theta ,\, \sinit )$, we define the reachability
function~$\varrho : S \to \FExp$ for~$\calS$ to be such that
\begin{displaymath}
  \forall P \in \calP \colon
  P \models \varrho(s) \ \ \text{iff} \ \
  \begin{array}[t]{@{}l}
    \exists \mkern1mu n
    \exists \mkern1mu s_0 \cdots s_n
    \exists \mkern1mu \alpha_1 \cdots \alpha_n
    \exists \mkern1mu \psi_1 \cdots \psi_n
    \colon
    \\
    s_0 = \sinit \land
    ( \forall i , 1 \leqslant i \leqslant n : 
    s_{i-1} \ftrans{\alpha_i}{\psi_i} s_{i} \land
    P \models \psi_i ) \land
    s_n = s
  \end{array}
\end{displaymath}
for all $s \in S$.  
We note that, for the ease of presentation in this paper, the
definition of an FTS above is slightly more abstract compared to the
original definition given in~\cite{CHSLR10:icse}.

Next, we introduce a notion of branching feature bisimulation for FTS,
generalizing the notion of branching bisimulation given by
Definition~\ref{def-branching-bisim} for LTS\@.

\blankline

\begin{definition}
  \label{df-branching-feature-bisim}
  Let $\calS = ( S ,\, \theta ,\, \sinit )$ and $\calS' = ( S' ,\,
  \theta' ,\, \sinit' )$ be two FTS\@.
  \begin{itemize}
  \item [(a)] For $s,s' \in S$, and satisfiable~$\eta \in \FExp$, we
    write $s \ftranss{\eta} s'$ if\
    \begin{math}
      \exists \mkern1mu n \,
      \exists \mkern1mu s_0 , \ldots , s_n \,
      \exists \mkern1mu \eta_1 , \ldots , \eta_n \colon 
      s = s_0 \land\
      \forall \mkern1mu i, 1 \leqslant i \leqslant n \colon 
      s_{i-1} \ftrans{\tau}{\eta_i} s_{i} \land
      s' = s_n \land 
      \eta = \textstyle{\bigwedge_{1 \leqslant i \leqslant n}} \:
      \eta_i .
    \end{math}
    Furthermore, we write $s \ftrans{(\alpha}{\psi)} s'$ in case $s
    \ftrans{\alpha}{\psi} s'$ or $\alpha = \tau \land s = s' \land
    \psi = \TRUE$.
  \item [(b)] A symmetric relation $R \subseteq S \times \FExphat
    \times S$ is called a branching feature bisimulation relation
    for~$\calS$ if for $s,t \in S$, $\alpha \in \calAtau$ such that
    $R(s ,\, \varphihat ,\, t)$ the so-called transfer condition
    holds:
    \begin{displaymath}
      \def\arraystretch{1.2}
      \begin{array}{l}
        s \ftrans{\alpha}{\psi} s' 
        \quad \textit{implies} \quad
        \exists \mkern1mu n \,
        \exists \mkern2mu \that_1 , \ldots , \that_n \,
        \exists \mkern2mu t'_1 , \ldots , t'_n \,
        \exists \mkern1mu \eta_1 , \ldots , \eta_n \,
        \exists \mkern1mu \psi_1 , \ldots , \psi_n \,
        \exists \mkern1mu \varphi_1 , \ldots , \varphi_n \,
        \exists \mkern1mu \varphi'_1 , \ldots , \varphi'_n 
        \colon
        {} \\ \qquad\qquad\qquad\qquad\qquad
        \forall \mkern1mu i, 1 \leqslant i \leqslant n \colon
        t \ftranss{\eta_i} \that_i \ftrans{(\alpha}{\psi_i)} t'_i \land
        R(s,\varphihat_i,\that_i) \land R(s',\varphihat'_i,t'_i) 
        \ \text{ and }
        {} \\ \qquad\qquad\qquad\qquad\qquad\qquad
        \forall P \in \calP \colon \ 
        P \models \varphi \land \psi 
        \  \Rightarrow \ 
        P \models
        \textstyle{\bigvee_{1 \leqslant i \leqslant n}} \: 
        \eta_i \land \psi_i \land \varphi_i \land \varphi'_i
      \end{array}
      \def\arraystretch{1.0}
    \end{displaymath}
  \item [(c)] Two states $s,t \in S$ are called branching feature
    bisimilar with respect to~$\calS$ if $R(s ,\,\TRUEhat ,\, t)$ for
    some branching feature bisimulation~$R$ for~$\calS$. Notation $s
    \bfbisim t$.
  \item [(d)] A branching feature bisimulation relation~$R$
    for~$\calS$ and $\calS'$ is called coherent if $R(s, \varphihat,
    s')$ implies $\varrho(s) \Rightarrow \varphi$, for all $s \in S$,
    $\varphi \in \FExp$, and~$s' \in S'$. Notation $\calS \cbfbisim
    \calS'$.
  \end{itemize}
\end{definition}

\blankline

\noindent
The specific subset of coherent branching feature bisimulations will
be used as a yardstick of comparison in the minimization algorithm
discussed in Section~\ref{sec-algorithm}. Intuitively, the
  feature expression~$\varrho(s)$ captures all products that can reach
  state~$s$. Coherency requires that $\varphi$~does not exclude part
  of these products. So the \lq products of~$s$\rq{} are not split
  by~$\varphi$, but treated as a coherent set of products.

Figure~\ref{fig-transfer-diagram-branching-feature-bisim} depicts the
general situation for the transfer condition where a transition $s
\ftrans{\alpha}{\psi} s'$ is matched by $n$~transition sequences
from~$t$ in total, viz.\ $t \ftranss{\eta_1} \that_1
\ftrans{(\alpha}{\psi_1)} t'_1$ to $t \ftranss{\eta_n} \that_n
\ftrans{(\alpha}{\psi_n)} t'_n$. Moreover, for a product~$P$ for which
state~$s$ admits the transition labelled~$\alpha$, i.e.\ a product
satisfying the constraint~$\varphi$ derived from~$R$ as well as the
feature expression~$\psi$ derived from the transition, it is required
that state~$t$ provides a related transition sequence labeled~$\alpha$
for this product as well. Thus, for some~$i$, $1 \leqslant i \leqslant
n$, $P$~meets $\eta_i$ and~$\psi_i$, thus can move from~$t$
to~$\that_i$ and~$t'_i$, while $P$ is included by the
constraint~$\varphi_i$ for the relation on~$s$ and~$\that_i$ and by
the constraint~$\varphi'_i$ on~$s'$ and~$t'_i$.

\begin{figure}[h]
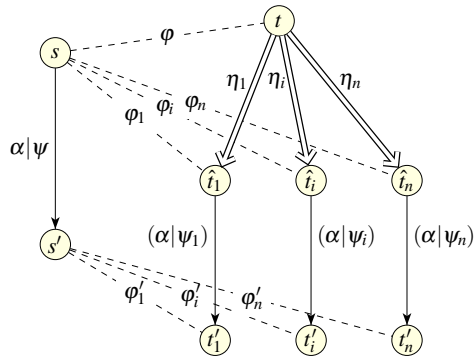

  \centering
  \small
  \scalebox{0.85}{%
  \begin{digraph}(72.5,52.5)(65,5)
    \graphset{Nadjust=n,Nw=4.75,Nh=4.75}
    \graphset{AHLength=2,AHangle=20}
    \graphset{iangle=180,fangle=270} 
    \node[fillcolor=LightYellow](s2)(70,50){$s$}
    \node[fillcolor=LightYellow](s'2)(70,20){$s'$}
    \node[fillcolor=LightYellow](t2)(105,55){$t$}
    \node[fillcolor=LightYellow](t1')(95,30){$\that_1$}
    \node[Nframe=n](t1x')(95.35,30.45){}
    \node[fillcolor=LightYellow](t1'')(95, 5){$t_1'$}
    \node[fillcolor=LightYellow](ti')(110,30){$\that_i$}
    \node[Nframe=n](tix')(109.9,30.7){}
    \node[fillcolor=LightYellow](ti'')(110, 5){$t_i'$}
    \node[fillcolor=LightYellow](tn')(125,30){$\that_n$}
    \node[Nframe=n](tnx')(124.8,30.7){}
    \node[fillcolor=LightYellow](tn'')(125, 5){$t_n'$}
          		
    \gasset{ELdistC=y,ELdist=0,dash={1}{0},AHnb=0}
    \edge(s2,t2){\colorbox{white}{$\varphi$}}
    \edge[curvedepth=0](s2,t1'){\colorbox{white}{$\varphi_1$}}
    \edge[ELside=r,curvedepth=0](s'2,t1''){\colorbox{white}{$\varphi_1'$}}
    \edge[ELpos=43,curvedepth=0](s2,ti'){\colorbox{white}{$\varphi_i$}}
    \edge[ELside=r,ELpos=53](s'2,ti''){\colorbox{white}{$\varphi_i'$}} 
    \edge[ELpos=40,curvedepth=0](s2,tn'){\colorbox{white}{$\varphi_n$}}
    \edge[ELside=r,ELpos=56,curvedepth=0](s'2,tn''){\colorbox{white}{$\varphi_n'$}}

    \gasset{ELdistC=n,ELdist=1,dash={0}1,AHnb=1}
    \edge[ELside=r](s2,s'2){$\alpha | \psi$}
    \edge[ELside=r,AHLength=2,AHlength=0,AHangle=45,AHnb=1,ELdist=0.5,ELpos=43,linewidth=1.0](t2,t1x'){$\eta_1$}
    \edge[ELside=r,AHLength=2,AHlength=0,AHangle=45,AHnb=1,linewidth=0.6,linecolor=white](t2,t1x'){}
    \edge[ELside=r,ELpos=35](t1',t1''){%
      $( \mkern-1mu \alpha | \psi_1 \mkern-1mu)$}
    \edge[ELside=r,AHLength=2,AHlength=0,AHangle=45,AHnb=1,ELdist=0.5,ELpos=37,linewidth=1.0](t2,tix'){$\eta_i$}
    \edge[ELside=r,AHLength=2,AHlength=0,AHangle=45,AHnb=1,linewidth=0.6,linecolor=white](t2,tix'){}
    \edge[ELside=l,ELpos=35](ti',ti''){%
      $( \mkern-1mu \alpha | \psi_i \mkern-1mu )$}
    \edge[ELside=l,ELpos=50,AHLength=2,AHlength=0,AHangle=45,AHnb=1,ELdist=0.5,ELpos=46,linewidth=1.0](t2,tnx'){$\eta_n$}
    \edge[ELside=r,AHLength=2,AHlength=0,AHangle=45,AHnb=1,linewidth=0.6,linecolor=white](t2,tnx'){}    
    \edge[ELside=l,ELpos=35](tn',tn''){%
      $( \mkern-1mu \alpha | \psi_n \mkern-1mu )$}
  \end{digraph}
  } 
  \caption{\label{fig-transfer-diagram-branching-feature-bisim}Transfer diagram for branching feature bisimilarity} 
\end{figure}

\noindent
Figure~\ref{fig-ex-strong-feature-bisim} below shows an example of two
FTS (without $\tau$-moves) at the left-hand side. At first sight the
relation $R = \lc (s_0,\TRUEhat,t_0)$, $(s_1,\varphihat_1,t_1)$,
$(s_2,\varphihat_2,t_1)$, $(s_3,\TRUEhat,t_2) \rc$ may look like a
branching feature bisimulation. However, a closer inspection of the
transition $t_1\ftrans{b}{(\varphi_1 \land\, \psi_1) \lor (\varphi_2
  \land\, \psi_2)} t_2$ reveals that this means that we need the
formulas $\varphi_i \land ((\varphi_1 \land \psi_1) \lor (\varphi_2
\land \psi_2)) \Rightarrow \psi_i \land \TRUE$ to hold for $i =
1,2$. However, this only holds when $\varphi_1 \land \varphi_2
\Rightarrow (\psi_1 \Leftrightarrow \psi_2)$; in that case $R$~is
indeed a branching feature bisimulation. Reversely, if a product meets
$\varphi_1 \land \varphi_2 \land \psi_1 \land \neg \mkern2mu \psi_2$,
there will be a transition for~$t_1$ for that product, but not
for~$s_2$ as shown by the two LTS at the right-hand side of
Figure~\ref{fig-ex-strong-feature-bisim}. It is clear that with a
transition from state~$s_0$ to state~$s_2$ but without a transition
between states $s_2$ and~$s_3$, on the one hand, and with a path from
$t_0$ to~$t_2$, on the other hand, the underlying LTS for the two FTS
(and therefore the FTS themselves as we shall see) cannot be
bisimilar.

\begin{figure}[h]
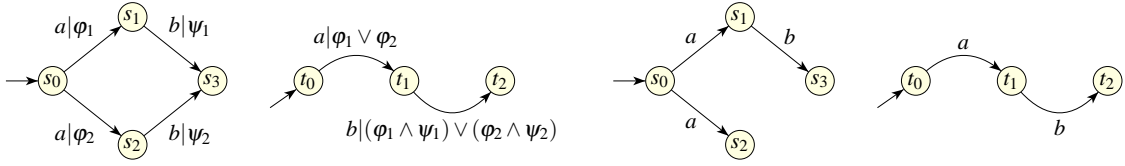

  \centering
  \small
  \scalebox{0.85}{%
  \begin{digraph}(170,22.5)(22.5,30)
    \graphset{Nadjust=n,Nw=4.5,Nh=4.5}
    \graphset{AHLength=2,AHangle=20}
    \graphset{iangle=180,fangle=270}
    \node[fillcolor=LightYellow](s)(25,40){$s_0$}
    \node[fillcolor=LightYellow](t)(37.5,50){$s_1$}
    \node[fillcolor=LightYellow](t')(37.5,30){$s_2$}
    \node[fillcolor=LightYellow](u)(50,40){$s_3$}
    \imark(s)
    \node[fillcolor=LightYellow](s2)(65,40){$t_0$}
    \node[fillcolor=LightYellow](t2)(80,40){$t_1$}
    \node[fillcolor=LightYellow](u2)(95,40){$t_2$}
    \imark[iangle=215](s2)
    \edge[ELside=l](s,t){$a | \varphi_1$}
    \edge[ELside=r](s,t'){$a | \varphi_2$}
    \edge[ELside=l](t,u){$b | \psi_1$}
    \edge[ELside=r](t',u){$b | \psi_2$}
    \edge[ELside=l,curvedepth=4](s2,t2){$a | {\varphi_1 \lor \varphi_2}$}
    \edge[ELside=r,curvedepth=-5](t2,u2){$b | {(\varphi_1 \land \psi_1) \lor
        (\varphi_2 \land \psi_2)}$}
    \put(95,0){%
    \node[fillcolor=LightYellow](s)(25,40){$s_0$}
    \node[fillcolor=LightYellow](t)(37.5,50){$s_1$}
    \node[fillcolor=LightYellow](t')(37.5,30){$s_2$}
    \node[fillcolor=LightYellow](u)(50,40){$s_3$}
    \imark(s)
    \node[fillcolor=LightYellow](s2)(65,40){$t_0$}
    \node[fillcolor=LightYellow](t2)(80,40){$t_1$}
    \node[fillcolor=LightYellow](u2)(95,40){$t_2$}
    \imark[iangle=215](s2)
    \edge[ELside=l](s,t){$a$}
    \edge[ELside=r](s,t'){$a$}
    \edge[ELside=l](t,u){$b$}
    \edge[ELside=l,curvedepth=4](s2,t2){$a$}
    \edge[ELside=r,curvedepth=-5](t2,u2){$b$}
    } 
  \end{digraph}
  } 
  \caption{\label{fig-ex-strong-feature-bisim}Bisimilar FTS assuming
    $\varphi_1 \land \varphi_2 \: \Rightarrow \: (\psi_1
    \Leftrightarrow \psi_2)$ and non-bisimilar LTS}
\end{figure}

\noindent
For branching feature bisimulation we have a strict 
correspondence with branching bisimulation for all products 
using the notion of a projection of an FTS. The projection 
results in an~LTS\@.

\blankline

\begin{definition}    
  Given an FTS~$\calS = ( S ,\, \theta ,\, \sinit )$ and a product~$P
  \in \calP$, the projection~$\calS_P$ of~$\calS$ for the product~$P$
  is the LTS $\calS_P = ( S ,\, {\transP{}} ,\, \sinit )$, where $s
  \transP{\alpha} s'$ if some~$\psi \in \FExp$ exists such that $s
  \ftrans{\alpha}{\psi} s'$ and $P \models \psi$, for $s,s' \in S$ and
  $\alpha \in \calAtau$.
\end{definition}

\blankline

\noindent
We use $s \bbisimP t$ to denote that $s$ and~$t$ are branching
bisimilar states for the projected LTS~$\calS_P$.

\blankline

\begin{theorem}
  \label{th-projection}
  Let~$\calS$ be an FTS with states $s$ and~$t$. It holds that $s
  \bfbisim t$ iff $s \bbisimP t$ for all $P \in \calP$.
\end{theorem}
\begin{proof}
  Suppose $R \subseteq { S \times \FExphat \times S }$ is a branching
  feature bisimulation relation with $R( s , \TRUEhat , t )$. Pick $P
  \in \calP$. Define $R_P = \lc (s',t') \mid \exists \mkern1mu \varphi
  \colon R( s' , \varphihat , t' ) \land P \models \varphi \rc$. We
  claim that~$R_P$ is a branching bisimulation relation
  with~$R_P(s,t)$.
  Clearly $R_P$ is symmetric and $R_P(s,t)$, since $R( s , \TRUEhat ,
  t )$ and $P \models \TRUE$. In order to verify the transfer
  condition for~$R_P$, suppose $R_P(s',t')$ and $s' \transP{\alpha}
  s''$. Pick, with appeal to the definitions of $R_P$ and~$\calS_P$,
  feature expressions~$\varphi, \psi$ such that (i)~$R( s' ,
  \varphihat , t' )$ and $P \models \varphi$, and (ii)~$s'
  \ftrans{\alpha}{\psi} s''$ and $P \models \psi$. Since $R$~is a
  branching feature bisimulation, we can find $\that_i$, $t_i$,
  $\eta_i$, $\psi_i$, $\varphi_i$ and~$\varphi'_i$, for $i = 1 ,
  \ldots , n$, such that
  \begin{displaymath}
    t' \fTrans{\eta_i} \that_i \ftrans{(\alpha}{\psi_i)} t'_i 
    \land
    R( s' ,\, \varphihat_i ,\, \that_i ) ,\,
    R( s'' ,\, \varphihat'_i ,\, t'_i )
    \quad \text{and} \quad
    P \: \models \:
    \textstyle{\bigvee_{1 \leqslant i \leqslant n}} \: 
    \eta_i \land \psi_i \land \varphi_i \land \varphi'_i
  \end{displaymath}
  for $i = 1, \ldots ,n$. Choose~$i$ such that $P \models \eta_i \land
  \psi_i \land \varphi_i \land \varphi'_i$. Since $t' \fTrans{\eta_i}
  \that_i \ftrans{(\alpha}{\psi_i)} t'_i$, $P \models \eta_i \land
  \psi_i \land \varphi_i$, and~$R( s'' , \varphihat'_i , t'_i )$, we
  have by definition of~$\calS_P$ and~$R_P$ that $t' \Trans \that
  \transP{(\alpha)} t''_i$ and $R_P(s'' , t'_i)$. Thus,
  $R_P$~satisfies the transfer condition, as was to be shown.

  To prove the reverse implication, pick for each $P \in \calP$, a
  branching bisimulation relation~$R_P$ such that~$R_P(s,t)$. Define
  $R \subseteq {S \times \FExp \times S}$ by
  \begin{math}
    R = 
    \lc ( s' , \varphihat , t' ) \mid \forall \mkern1mu P \in \calP
   \colon P \models \varphi \Leftrightarrow R_P(s',t')
   \rc 
  \end{math}.
  We verify that $R$ is a branching feature bisimulation. Clearly,
  $R(s ,\, \TRUEhat ,\, t)$.
  In order to check the transfer condition for~$R$, suppose $R(s' ,
  \varphihat , t' )$ and $s' \ftrans{\alpha}{\psi} s''$. Then it
  holds, for all $P \in \calP$ with $P \models \varphi$, that
  $R_P(s',t')$. Moreover, for all $P \in \calP$ with $P \models
  \psi$, we have $s' \transP{\alpha} s''$. Thus, for all $P \in \calP$
  with $P \models \varphi \land \psi$, we can pick $\that_P, t'_P$
  and~$\eta_P ,\psi_P$ such that $P \models \eta_P \land \psi_P$, $t'
  \fTrans{\eta_P} \that_P \ftrans{(\alpha}{\psi_P)} t'_P$ and $R_P(s',
  \that_P)$ and~$R_P(s'', t'_P)$.

  Suppose $\lc P \in \calP \mid P \models \varphi \land \psi \rc = \lc
  P_1 ,\, \ldots ,\, P_k \rc$. Also, for $i = 1, \ldots , k$, let
  $\that_i, t'_i$ and~$\eta_i, \psi_i$ be shorthand for $\that_{P_i},
  t'_{P_i}$ and~$\eta_{P_i}, \psi_{P_i}$, respectively. Since $P_i
  \models \chi(P_i)$, $R_{P_i}(s',\that_i)$ and $R_{P_i}(s'',t'_i)$,
  it holds that $R( s' ,\, \varphihat_i ,\, \that_i )$ and $R(
  s'' ,\, \varphihat'_i ,\, t'_i )$ 
  for $\varphi_i , \varphi'_i \in \FExp$ such that
  $\widehat{\chi(P_i)} \Rightarrow \varphi_i$ and $\widehat{\chi(P_i)}
  \Rightarrow \varphi'_i$. 
  We conclude that, for $i = 1 , \ldots , k$,
  it holds that
  $t' \fTrans{\eta_i} \that_i \ftrans{(\alpha}{\psi_i)} t''_i$, 
  $R( s' , \varphihat_i , \that_i )$ and 
  $R( s'' , \varphihat'_i  , t'_i )$ 
  while $P \models \varphi \land \psi$ $\Rightarrow$ $P \models
  \textstyle{\bigvee_{1 \leqslant i \leqslant n}} \: \eta_i \land
  \psi_i \land \varphi \land \varphi'$, which verifies the transfer
  condition for~$R$.
\end{proof}


\noindent
The theorem asserts the soundness of branching feature bisimulation
for FTS with respect to branching bisimulation for the projected LTS
for all products. In the sequel, we propose an algorithm for
minimization of an FTS modulo branching feature bisimulation and
compare, in a case study, verification of properties against the
minimized FTS to verification of properties against the minimized
product LTS\@.


\section{Minimization modulo coherent branching feature bisimulation} 
\label{sec-algorithm}

When minimizing an FTS~$\calS$ we look for an FTS~$\calS'$ satisfying
$\calS \bfbisim \calS'$ and such that it is the smallest in \lq
size\rq. For branching bisimulation for LTS it is the case that a
branching bisimilar LTS with the minimal number of states also has the
minimal number of transitions (after removal of
$\tau$-loops). Algorithms for branching bisimulation reduction make
use of this fact by looking for the unique LTS with the minimal number
of states. Unfortunately, this is not true for branching feature
bisimulation, as is demonstrated in
Figure~\ref{fig:notMinTransitions}: The FTS $\calT$ and~$\calU$ are
both branching feature bisimilar to FTS~$\calS$, and both have the
minimal number of states. However, $\calU$~has twice as many
transitions as~$\calT$.\\

\begin{figure}[h]
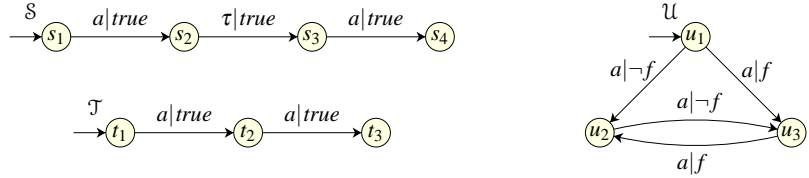

  \begin{center}
    \small
    \scalebox{0.85}{%
    \begin{digraph}(125,16)(-15,29.5)
      \graphset{iangle=180,fangle=270} 
      \graphset{Nadjust=n,Nw=4.5,Nh=4.5}
      \graphset{AHLength=1.6,AHlength=1,AHangle=30}
      \node[Nframe=n](S)(-14,49){$\calS$}
      \node[fillcolor=LightYellow](s)(-10,45){$s_1$}
      \node[fillcolor=LightYellow](s2)(10,45){$s_2$}
      \node[fillcolor=LightYellow](s3)(30,45){$s_3$}
      \node[fillcolor=LightYellow](s4)(50,45){$s_4$}
      \imark(s)
      \node[Nframe=n](T)(-4,34){$\calT$}
      \node[fillcolor=LightYellow](t)(0,30){$t_1$}
      \node[fillcolor=LightYellow](t2)(20,30){$t_2$}
      \node[fillcolor=LightYellow](t3)(40,30){$t_3$}
      \imark(t)
      \node[Nframe=n](U)(86,49){$\calU$}
      \node[fillcolor=LightYellow](u2)(75,30){$u_2$}
      \node[fillcolor=LightYellow](u)(90,45){$u_1$}
      \node[fillcolor=LightYellow](u3)(105,30){$u_3$}
      \imark(u)
      \edge[ELside=l](s,s2){$a | \mkern1mu \TRUE$}
      \edge[ELside=l](s2,s3){$\tau | \mkern1mu \TRUE$}
      \edge[ELside=l](s3,s4){$a | \mkern1mu \TRUE$}
      \edge[ELside=l](t,t2){$a | \mkern1mu \TRUE$}
      \edge[ELside=l](t2,t3){$a | \mkern1mu \TRUE$}
      \edge[ELside=l,ELdist=0](u,u3){$a | f$}
      \edge[ELside=r,ELdist=-0.5,ELpos=50](u,u2){$a | \neg f$}
      \edge[ELpos=52, ELside=l, curvedepth=2](u2,u3){$a | \neg f$}
      \edge[ELpos=48, ELside=l, curvedepth=2](u3,u2){$a | f \phantom{\neg}$}
    \end{digraph}
    } 
\end{center}
\caption{\label{fig:notMinTransitions}Three branching feature
  bisimilar FTS}
\end{figure}

\noindent
We see that the property of feature bisimulation that allows to merge
multiple transitions with the same label and different feature
expressions into a single transition now hinders us, since it also
allows to split transitions. To avoid this problem we restrict to
\emph{coherent} bisimulations
(cf.\ Definition~\ref{df-branching-feature-bisim}d). Thus, we require
that states of~$\calS$ can only be related to states of the
reduced~$\calS'$ for (supersets of) their reachability
set. Unfortunately, this recipe does not guarantee that a minimal FTS
is found, as Figure~\ref{fig:nonsplitNotMinimal} below shows, but
among all coherent branching feature bisimilar FTS our algorithm is
able to find the smallest one, see
  Theorem~\ref{th-correctness-algorithm}.

\begin{figure}[h]
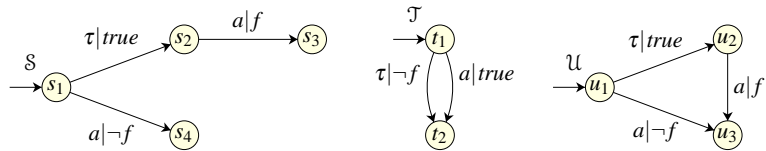

  \begin{center}
    \small
    \scalebox{0.85}{%
    \begin{digraph}(120,17)(-12,33)
      \graphset{iangle=180,fangle=270} 
      \graphset{Nadjust=n,Nw=4.5,Nh=4.5}
      \graphset{AHLength=1.6,AHlength=1,AHangle=30}
      \node[Nframe=n](S)(-14,41.5){$\calS$}
      \node[fillcolor=LightYellow](s)(-10,37.5){$s_1$}
      \node[fillcolor=LightYellow](s2)(10,45){$s_2$}
      \node[fillcolor=LightYellow](s3)(30,45){$s_3$}
      \node[fillcolor=LightYellow](s4)(10,30){$s_4$}
      \imark(s)
      \node[Nframe=n](T)(46,49){$\calT$}
      \node[fillcolor=LightYellow](t)(50,45){$t_1$}
      \node[fillcolor=LightYellow](t2)(50,30){$t_2$}
      \imark(t)
      \node[Nframe=n](U)(71,41.5){$\calU$}
      \node[fillcolor=LightYellow](u2)(95,30){$u_3$}
      \node[fillcolor=LightYellow](u)(75,37.5){$u_1$}
      \node[fillcolor=LightYellow](u3)(95,45){$u_2$}
      \imark(u)
      \edge[ELside=l](s,s2){$\tau | \mkern1mu \TRUE$}
      \edge[ELside=l](s2,s3){$a | f$}
      \edge[ELside=r](s,s4){$a | \neg f$}
      \edge[ELpos=40,ELside=l,curvedepth=2](t,t2){$a | \mkern1mu \TRUE$}
      \edge[ELpos=40, ELside=r, curvedepth=-2](t,t2){$\tau | \neg f$}      
      \edge[ELside=l,ELdist=0.5,ELpos=50](u,u3){$\tau | \mkern1mu \TRUE$}
      \edge[ELside=r,ELdist=0.5,ELpos=50](u,u2){$a | \neg f$}
      \edge[ELside=l](u3,u2){$a | f$}
    \end{digraph}
    } 
  \end{center}
  \caption{\label{fig:nonsplitNotMinimal}Minimal branching feature
    bisimilar vs.\ minimal coherent branching feature bisimilar}
\end{figure}

\noindent
In Figure~\ref{fig:nonsplitNotMinimal}, FTS~$\calT$ is branching
feature bisimilar to FTS~$\calS$, and has the minimal number of states
and transitions. However, when restricting to coherent branching
feature bisimulation relations, FTS~$\calU$ is the smallest FTS that
can be obtained from~$\calS$ such that $\calS \cbfbisim \calU$. Note
that the relation~$R$ with $R(s_2,\hat{f},t_1)$ and $R(s_2, \neg
\mkern-2mu \hat{f},t_2)$ is not coherent, since $\varrho(s_2) = \TRUE$
does not imply $f$ nor~$\neg \mkern-2mu f$.  We will adapt the
reduction algorithm described in Section~\ref{sec-branching-bisim} for
minimization modulo coherent branching feature bisimulation.

Before describing the algorithm, we first show that the problem of
coherent branching feature bisimulation minimization is NP-hard by
reducing the chromatic number problem to it: given a graph, what is
the minimum number of colors to color the nodes such that adjacent
nodes have different colors? To verify the construction, we need an
auxiliary result.


\begin{lemma}
  \label{thm:reachBisimToBisim}
  Let~$\calS = ( S ,\, \theta ,\, s_0 )$ be an FTS with states $s$
  and~$t$. If $R(s ,\, \varrho(s) {\land} \mkern1mu \varrho(t) ,\, t)$
  for a branching feature bisimulation relation~$R$, then $\calS
  \cbfbisim \calS'$ with states $s$ and~$t$ related to a single state
  of~$\calS'$. 
\end{lemma}
\begin{proof}
  Let $\calS' = ( S' ,\, \theta' ,\, \sinit' )$ with $S' = ( S
  {\setminus} \singleton{s, t} ) \cup \singleton{r}$ for some $r
  \notin S$, with $\theta'(u, a, v) = \theta(u, a, v)$ for~$u, v \in
  S'$, $u,v \neq r$ and $\theta'(u, a, r) = \theta(u, a, s) \lor
  \theta(u, a, t)$ for~$u \neq r$, $\theta'(r, a, v) = \theta(s, a, v)
  \lor \theta(t, a, v)$, for~$v \neq r$, and $\theta'(r, a, r) =
  \bigvee_{q, w \in \singleton{s,t}} \theta(q, a, w)$, and finally
  with $\sinit' = \sinit$ if~$\sinit \neq s,t$, and $\sinit' = r$
  otherwise.
  Using that $R$ is a branching feature bisimulation with~$R(s,
  \varrho(s) \land \varrho(t), t)$, one constructs a coherent
  branching feature bisimulation~$R'$ such that $R'(s,\varrho(s),r)$
  and~$R'(t,\varrho(t),r)$.
\end{proof}

\blankline

\noindent
Next we set the stage for a reduction of graph coloring to coherent
branching feature bisimulation minimization.
Consider an undirected graph $\calG = (V, E)$ with nodes in~$V$ and
edges in~$E$. Let~$\calA = \singleton{a}$, $\calF = \lc f_v \mid v \in
V \rc$ and $\calP = \lc P_v \mid v \in V \rc$. The FTS $\calS_G = (
S_G ,\, \theta_G ,\, s_1 )$ of~$\calG$ is such that $S_G = \lbrace s_1
,\, s_2 \rbrace \cup \lc s_v \mid v \in V \rc$ for distinct
states~$s_1$ and $s_2$, $\theta(s_1, a, v) = \bigvee_{u \in V} \: \lc f_u
\mid (u, v) \in \calG \rc \lor f_v$ for all $v \in V$, and $\theta(v,
a, s_2) = f_v$, and finally such that $\theta(s, a, s') = \FALSE$ in
all other cases.

\blankline

\begin{theorem}
  \label{thm:algo-np-hard}
  Let $\calS'_G$ be the minimal FTS that is coherent branching feature
  bisimilar to the FTS~$\calS_G$ given above. Then the number of
  states in~$\calS'_G$ is equal to the chromatic number of~$\calG$
  plus~$2$.
\end{theorem}
\begin{proof}
  Let $\Gamma$ be a set of colors. Suppose $\gamma : V \to \Gamma$ is
  a coloring of~$\calG$ using all colors. Then the FTS $( \lbrace s_1
  , s_2 \rbrace \cup \Gamma ,\, \theta_\gamma ,\, s_1 )$, where
  $\theta_\gamma (s_1,a,C) = \bigvee_{\gamma(u) = C} \:
  \theta(s_1,a,s_u)$, $\theta_\gamma (C,a,s_2) = \bigvee_{\gamma(u) =
    C} \: f_u$ is coherent bran\-ching feature bisimilar to~$\calS_G$
  via the relation $R$ such that $R(s_i, \TRUEhat, s_i)$ for $i =
  1,2$, and $R(s_u, \varrho(s_u), \gamma(u))$.

  Reversely, an FTS~$\calS'$ that is coherent branching feature bisimilar
  to~$\calS_G$ can only identify states $s_u,s_v$ for $u,v \in
  V$. Hence such an FTS induces a coloring for~$\calG$: Pick for each
  state~$s_v$ a single~$s' \in S'$ such that $R(s_v,\varphi,s')$ for a
  coherent branching feature bisimulation~$R$ relating $\calS$
  and~$\calS'$.  If states $s_u$ and~$s_v$ correspond to the same
  state of~$\calS'$, there can be no edge between $u$ and~$v$
  in~$\calG$. For if $(u,v)$ is an edge in~$\calG$, we have $s_1
  \trans{a}{} u \trans{a}{} s_2$ and $s_1 \trans{a}{} v \nrightarrow
  {}$ in the projection of~$\calS_G$ for the product~$p_u$, but $s_1
  \trans{a}{} u \nrightarrow {}$ and $s_1 \trans{a}{} v \trans{a}{}
  s_2$ in the projection of~$\calS_G$ for the product~$p_v$.

  It follows that the FTS~$\calS'_G$ that is minimal coherent
  branching feature bisimilar to~$\calS_G$ corresponds to a minimal
  coloring of~$\calG$. Moreover, the number of states different from the
  images of $s_1$ and~$s_2$ corresponds to the number of colors
  needed.
\end{proof}

\blankline

\noindent
Note how, in the proof above, the coherence condition \lq if
$R(s,\varphi,s')$ then $\varrho(s) \Rightarrow \varphi$\rq\ enforces that
for the minimal FTS~$\calS'_G$ the products that can reach~$s$
in~$\calS_G$ are not split over multiple states in~$\calS'_G$.
From the theorem we obtain the following result.

\blankline

\begin{corollary}
  Constructing a minimal coherent branching feature bisimilar FTS
  is NP-complete.
  \qed
\end{corollary}

\blankline

\noindent 
Before we provide an algorithm for minimization of an FTS modulo
coherent branching feature bisimulation, we slightly generalize the
notion of a partition as used in Section~\ref{sec-branching-bisim}, to
allow a state to belong to separate groups of products.

A collection $\calB = \lc B_i \mid i \in I \rc$ of non-empty subsets
of a set~$S$ is called a \emph{semi-partition} of~$S$ if
(i)~$\bigcup_{i \in I} \: B_i = S$, and (ii)~for $j \neq i: B_j
\setminus B_i \neq \varnothing$. Thus, $\calB$ covers~$S$ and no~$B_j$
is strictly contained in a~$B_i$. Also, for a semi-partition its
elements are referred to as \emph{blocks}. We say that a
semi-partition~$\calB'$ is a refinement of a semi-partition~$\calB$ if
every block of~$\calB'$ is a subset of a block of~$\calB$. Likewise,
we say that $\calB$ is coarser than~$\calB'$. A semi-partition~$\calB$
of~$S$ induces a relation $\sim_{\calB}$ on~$S$ (not necessarily an
equivalence relation), where two elements of~$S$ are related iff they
are included in the same block of~$\calB$.

Given an FTS $\calS = ( S ,\, \theta ,\, \sinit )$, we first do some
preprocessing. We eliminate unreachable states and strengthen the
transition constraint with the reachability condition for its source
state:
\begin{displaymath}
  S \assign \lc s \in S \mid \varrho(s) \notsimP\, \FALSE \rc
  \quad \text{and} \quad
  \theta(s,\alpha,s') \assign \theta(s,\alpha,s') \land \varrho(s)
\end{displaymath} 
We define the set $\calAf$ of so-called featured labels by $\calAf =
\lc (\alpha, \psi) \mid \exists \mkern1mu s,t \, \exists \mkern1mu \alpha
\colon \theta(s,\alpha,t) =
\psi \land \psi \notsimP\, \FALSE \rc$.
For a semi-partition~$\calB$ of~$S$, $B, B' \in \calB$ and featured
label $(\alpha,\psi) \in \calAf$ we let
\begin{displaymath}
  \def\arraystretch{1.2}
  \begin{array}{l}
    \textit{non-neg}_{(\alpha,\psi)}(B,B') \; = \; 
    \lc s \in B \mid 
    \forall P \in \calP ,\, P \models \varrho(s) \land \psi \colon
    \exists \mkern1mu n \,
    \exists \mkern1mu s_0 , \ldots , s_n \in B \mkern1mu \,
    \exists \mkern1mu s' \in B' \,
    \exists \mkern1mu \psi_1 , \ldots , \psi_n , \psi'
    \colon 
    \\ \qquad \qquad \qquad \qquad \qquad
    s_0 = s \land {} 
    ( \forall \mkern1mu i, 1 \leqslant i \leqslant n \colon
    s_{i-1} \ftrans{\tau}{\psi_i} s_i \land P \models \psi_i ) \land 
    s_n \ftrans{(\alpha}{\psi')} s' \land P \models \psi' \rc,
  \end{array}
  \def\arraystretch{1.0}
\end{displaymath}
and define its subset $\textit{pos}_{(\alpha,\psi)}(B, B')$ to include
all $s \in \textit{non-neg}_{(\alpha,\psi)}(B, B')$ for which $\psi
\Rightarrow \varrho(s)$ and $s_n \ftrans{\alpha}{\psi'} s'$ for $s_n
\in B$, $s' \in B'$ as above. Moreover, we define
$\textit{neg}_{(\alpha,\psi)}(B, B') = B \setminus
\textit{non-neg}_{(\alpha,\psi)}(B, B')$.
We know for sure that two states $s$ and~$t$ of a block~$B$ are
behaviorally different, if $s \in \textit{pos}_{(\alpha,\psi)}(B,B')$
and $t \in \textit{neg}_{(\alpha,\psi)}(B,B')$.  Therefore, we say
that $B'$ is a \emph{splitter} of~$B$ with respect to~$(\alpha, \psi)$
if $B \neq B'$ or $\alpha \neq \tau$, and
$\mathit{pos}_{(\alpha,\psi)}(B,B'), \mathit{neg}_{(\alpha,\psi)}(B,
B') \neq \varnothing$ 
%
%
(meaning there is at least one state in the $\textit{pos}$-set
that must do an actual $\tau$-step to reach~$B'$).
If $\calB$ is a semi-partition of~$S$ and $B'$ is a splitter of $B$
with respect to $(\alpha,\psi)$, then the semi-partition~$\calB'$ is
obtained from~$\calB$ by replacing block~$B$ by $B_1=
\textit{non-neg}_{(\alpha,\psi)}(B,B')$ and $B_2 = B {\setminus} \,
\mathit{pos}_{(\alpha,\psi)}(B,B')$. However, in the case that $B_1$
or~$B_2$ is a subset of another block in the partition (apart
from~$B$), it is not added to ensure that~$\calB'$ is a
semi-partition.

The minimization algorithm starts from the trivial
semi-partition~$\singleton{S}$, and keeps refining the semi-partition
until no splitters are left. This results in the coarsest
semi-partition, but still a block may be covered completely by other
blocks. Therefore, as post-processing, we remove as many blocks as
possible from the semi-partition, while preserving the
semi-partition properties, to find the smallest semi-partition
(e.g.\ using an algorithm for the minimum set cover problem).
\begin{itemize}
  \item []
  $\calB \assign \lbrace S \rbrace$\,;
  \smallskip \\
  \textbf{while} a splitter~$B'$ for a block~$B$ with respect to a
  featured label~$(\alpha,\psi)$ exists \textbf{do} 
  \smallskip \\ \mbox{} \qquad
  $\calB \assign \calB {\setminus} \singleton{B}$\,;
  \smallskip \\ \mbox{} \qquad
  \textbf{if}
  $\textit{non-neg}_{(\alpha,\psi)}(B,B') \subseteq B''$ for no~$B''
  \in \calB$ \textbf{then} 
  $\calB \assign \calB \cup 
  \lbrace \textit{non-neg}_{(\alpha,\psi)}(B,B') \rbrace$
  \textbf{end}\,;
  \smallskip \\ \mbox{} \qquad
  \textbf{if}
  $B {\setminus} \, \textit{pos}_{(\alpha,\psi)}(B,B') \subseteq B''$ for no~$B''
  \in \calB$ \textbf{then} 
  $\calB \assign \calB \cup  
  \lbrace B {\setminus} \, \textit{pos}_{(\alpha,\psi)}(B,B') \rbrace$
  \smallskip
  \textbf{end} \\
  $\calBmin \assign \text{smallest subset of $\calB$ covering $S$}$\,;
\end{itemize}
It is easy to see that the algorithm terminates: Note that after
each iteration at least two states have been permanently split from
each other. Since there are less than $|S|^2$~possible pairs of states
in $S$, termination will occur in at most $|S|^2$~iterations. In the
theorem below, we call a semi-partition~$\calC$ a \emph{stable}
partition with respect to a block~$B'$ if for no block~$B$ and for no
featured label~$(\alpha,\psi)$, $B'$~is a splitter of~$B$ with
respect to~$(\alpha,\psi)$. The semi-partition~$\calC$ is itself called stable if $\calC$ is stable with respect to all its blocks.

\blankline

\begin{lemma}
  \label{thm:coarsestSemiPartition}
  For an FTS $\calS = ( S ,\, \theta ,\, \sinit )$, $\calBmin$
  obtained from the algorithm is the smallest stable semi-partition
  refining $\lbrace S \rbrace$.
\end{lemma}
\begin{proof}
  We show by induction on the number of iterations of the algorithm
  that each stable partition refines the current
  semi-partition~$\calB$.
  Let $\calC$ be a stable semi-partition. Clearly the statement holds
  initially, each semi-partition refines~$\singleton{S}$. Suppose
  $\calC$ refines semi-partition~$\calB$ obtained after a number of
  iterations and suppose a splitter~$B'$ of a block~$B$ exists with
  respect to a featured label~$(\alpha,\psi)$. It suffices to show
  that any block~$C$ of~$\calC$ is included in a block of~$\calB'$,
  the semi-partition obtained by splitting~$B$. Pick a block
  of~$\calB$ containing~$C$. If this block is different from~$B$, we
  are done. So, suppose $C \subseteq B$. We have to show that either
  $C \subseteq \textit{non-neg}_{(\alpha,\psi)}(B,B')$ or $C \subseteq
  B \setminus \mathit{pos}_{(\alpha,\psi)}(B,B')$.

  Suppose $s,t \in C$ with $s \in \mathit{pos}_{(\alpha,\psi)}(B,B')$
  and $t \in \mathit{neg}_{(\alpha,\psi)}(B,B')$. We derive a
  contradiction. Pick a product $P \in \calP$ such that $P \models
  \psi$. Such a product exists by definition of~$\calAf$. Choose $s_0
  , \ldots , s_n \in B$, $s' \in B'$, $\psi_1 , \ldots , \psi_n ,
  \psi' \in \FExp$ such that $s_0 = s$, $s_{i-1} \ftrans{\tau}{\psi_i}
  s_i$ for $1 \leqslant i \leqslant n$, $s_n \ftrans{(\alpha}{\psi')}
  s'$, and moreover $P \models \psi_i$, for $1 \leqslant i \leqslant
  n$, and $P \models \psi'$. Let $C_0, \ldots, C_n, C'$ be the blocks
  of~$\calC$ such that $s_i \in C_i$ and $s' \in C'$. Note that $C_i
  \subseteq B$, for $0 \leqslant i \leqslant n$, and $C' \subseteq
  B'$. Using the fact that $\calC$ is stable we can construct a
  sequence $t_0, \ldots, t_m \in B$, $t' \in B'$, $\varphi_1 , \ldots
  , \varphi_m , \varphi' \in \FExp$ such that $t_0 = t$, $t_{i-1}
  \ftrans{\tau}{\varphi_i} t_i$ for $1 \leqslant i \leqslant m$, $t_n
  \ftrans{(\alpha}{\varphi')} t'$, and moreover $P \models \varphi_i$
  for $1 \leqslant i \leqslant m$, and $P \models \varphi'$. This
  contradicts $t \in \mathit{neg}_{(\alpha, \psi)}(B, B')$, and proves
  the induction step.
  Finally, we observe that $\calBmin$ itself is a stable
  semi-partition that refines~$\singleton{S}$.
\end{proof}

\blankline

\begin{lemma}
  \label{thm:fbisimToSemiPartition}
  Let $\calS = ( S ,\, \theta ,\, \sinit)$ be an FTS, and $\calS' = (
  S' ,\, \theta' ,\, \sinit')$ be an FTS such that $\calS \cbfbisim
  \calS'$ by a relation~$R$. Then $R$ defines a stable semi-partition
  $\calC$ of~$S$ such that $s \sim_{\calC} t$ iff $\exists \mkern1mu r \in S'
  \colon R(s, \varrho(s), r) \land R(t, \varrho(t), r)$.
\end{lemma}
\begin{proof}
  We have to show that $\calC$ is stable indeed. Suppose that there
  are blocks $B, B'$ in $\calC$ such that $B'$ is a splitter of $B$
  with respect to a featured label~$(\alpha,\psi)$. This means there
  are states $s$ and $t$ in $B$ such that $s \in
  \mathit{pos}_{(\alpha,\psi)}(B,B')$ and $t \in
  \mathit{neg}_{(\alpha,\psi)}(B,B')$. We pick $P \in \calP$ such that
  $P \models \varrho(s) \land \varrho(t) \land \psi$. By definition of
  the ${\it pos}$-set there exist $s_0 , \ldots , s_n \in B$, $s' \in
  B'$, $\psi_1 , \ldots , \psi_n , \psi' \in \FExp$ such that $s_0 =
  s$, $s_{i-1} \ftrans{\tau}{\psi_i} s_i$ for $1 \leqslant i \leqslant
  n$, $s_n \ftrans{(\alpha}{\psi')} s'$, and moreover $P \models
  \psi_i$, for $1 \leqslant i \leqslant n$, and $P \models
  \psi'$. Since $s_n \in B$ we have, by construction of $\calC$, both
  $R(s_n,\varrho(s_n),r)$ and $R(t,\varrho(t),r)$ for suitable $r \in
  \calS'$.  Therefore, there exists a feature bisimulation relation
  $R'$ on $\calS$ such that $R'(s_n, \varrho(s_n) \land \varrho(t),
  t)$. Using the transfer condition of this relation we can construct
  a sequence $t_0, \ldots, t_m \in B$, $t' \in B'$, $\varphi_1 ,
  \ldots , \varphi_m , \varphi' \in \FExp$ such that $t_0 = t$,
  $t_{i-1} \ftrans{\tau}{\varphi_i} t_i$ for $1 \leqslant i \leqslant
  m$, $t_n \ftrans{(\alpha}{\varphi')} t'$, and moreover $P \models
  \varphi_i$ for $1 \leqslant i \leqslant m$, and $P \models
  \varphi'$. This contradicts $t \in \mathit{neg}_{(\alpha, \psi)}(B,
  B')$, and proves that $\calC$ is stable.
\end{proof}

\blankline

\noindent
We are now in a position to prove the correctness of the minimization
algorithm. 

\blankline

\begin{theorem}
  \label{th-correctness-algorithm}
  Assume that $\calB$ is the partition obtained upon termination after
  applying the algorithm to the FTS $\calS = ( S ,\, \theta ,\, \sinit )$. 
  Define the FTS $\calS_{\mathit{min}} = ( \calB ,\,
  {\theta_{\mathit{min}}} ,\, B_{\mkern-1mu \ast} )$ by letting
  (i)~$\theta_\mathit{min}(B,\alpha,B') = \bigvee \lc \theta(s,a,s')
  \mid s \in B ,\, s' \in B' \rc$ with $B \neq B'$ or $\alpha \neq
  \tau$, and (ii)~by choosing $B_{\mkern-1mu \ast}$ such that 
  $\sinit \in B_{\mkern-1mu \ast} $. Then $\calS_{\mathit{min}}$ is 
  the smallest FTS that is coherent branching feature bisimilar 
  to~$\calS$.
\end{theorem}
\begin{proof}
  By Lemma~\ref{thm:coarsestSemiPartition} we have that $\calBmin$ is
  the smallest stable semi-partition refining $\singleton{S}$. It
  suffices to show, using Lemma~\ref{thm:reachBisimToBisim}, that a
  coherent branching feature bisimu\-lation for $\calS$
  and~$\calS_{\mathit{min}}$ exists. Since, by
  Lemma~\ref{thm:fbisimToSemiPartition} we have that every coherent
  branching feature bisimulation relation from~$\calS$ to an
  FTS~$\calS'$ induces a stable semi-partition on~$\singleton{S}$,
  implying that $\calS_{\mathit{min}}$ is indeed minimal.
\end{proof}

\blankline

\noindent
Thus, given an FTS~$\calS$, we continue to refine the trivial
semi-partition until no more splitter can be found. Splitting a 
block is done cautiously: (i) it must eliminate a splitter and 
(ii) it must yield a semi-partition again. The final semi-partition 
that is reached induces an FTS~$\calSmin$ that is the smallest 
FTS that is coherent branching feature bisimilar to~$\calS$. The 
next section reports on a small case study using this approach.


\section{Experimental evaluation}
\label{sec-experiments}

We extended the example SPL of a coffee vending machine 
described in~\cite{ABFG11b,ABFG12,BV14b,BV14:formalise} 
with a soup component running in parallel. The complete SPL 
consists of~$18$ features and $118$ products and the FTS 
modeling it contains~$182$ states and~$691$ transitions. 
The details of this SPL can be found in Appendix~\ref{app:SPL}\@.
Basically, each product contains the well-known beverage component and
optionally a soup component, and allows the insertion of either euros
or dollars (returned upon a cancel) in either of its components. The
user chooses a beverage (sugared or not) among those offered (at least
coffee, cappuccino only for euros) or else a type of soup (at least
one among chicken, tomato, pea). The user must place a cup to get
soup. A cup detector is optional (mandatory for dollars). When
present, soup is only poured if a cup was placed, else soup may be
spilled.  Placing a cup may need to be repeated if not detected. A
soup order may be canceled until a cup is detected. Optionally, a
shared ringtone may ring after delivery (mandatory for cappuccino),
after which the user takes a cup (with a drink or soup) and can again
insert money in either component.
Concrete features have an associated cost (zero for abstract features)
and the total cost of a product, summing the costs of the features it
includes, does not exceed the fixed upper bound of~35.

We used the mCRL2 toolset to verify the~$12$ properties listed in
Appendix~\ref{app:SPL} against this SPL, both product-by-product and
by using the FTS-based family approach described
in~\cite{BV14b,BV14:formalise}, and both with and without branching
(feature) bisimulation minimization. For the approach with
bisimulation we applied branching feature bisimulation to the FTS,
resulting in a reduced FTS, which we projected to obtain the reduced
LTS for each product. The results are shown in
Table~\ref{tab:results}. For the product-by-product approaches,
generating the projections for all products is included in the
computation time, and so is the time for bisimulation
reduction in case of the approaches with bisimulation. To even out
effects caused by other processes running whilst performing the
experiments, all computation times are averaged over~$5$ runs.

Regarding the product-by-product approach, performing bisimulation
reduction for the product LTS reduces the computation time by
about~$8\%$. For property~2 (\emph{The SPL is deadlock-free}), the
computation time with bisimulation is significantly larger than for
other properties. In this case abstraction does not reduce the LTS\@.
A similar observation holds for properties~1 (\emph{If a coffee is
  ordered, it is eventually poured}), 5a (\emph{If a beverage is
  ordered, then eventually it is canceled or a cup is taken}) and~5b
(\emph{If soup is ordered, then eventually it is canceled, a cup is
  taken or the customer has bad luck}), which are false, but deemed
true after applying bisimulation reduction. They state that
some\-thing eventually happens, which is not true in reality since the
two components are running in parallel, thus abstraction creates
infinite loops that allow postponing that something
indefinitely. Applying bisimulation reduction causes these loops to be
abstracted from completely, making the properties true for the reduced
system. However, standard tricks, like the explicit signaling of the
end of a cycle, could be applied to alleviate this problem.

Now consider the FTS-based family approach. Without applying 
bisimulation reduction, the total computation time increases by 
almost~$50\%$ with respect to the product-by-product approach. 
Hence, for this SPL, FTS-based verification with mCRL2 is not 
beneficial compared to regular enumerative verification. 
However, if we apply bisimulation reduction, then the FTS-based
computation times decrease by~$>\!\!70\%$.  Only property~2 still
needs more computation time than in the product-based approach (again
because abstraction is not beneficial for the verification). Note that
in case less actions are involved in a property, it is possible to
abstract from larger parts of the FTS, implying faster
verification. This effect was much less in the product-by-product
approach. Hence, the more local a property, the more beneficial it is
to perform FTS-based family verification in combination with branching
feature bisimulation reduction using mCRL2\@. Obviously, this
observation needs to be confirmed by experimenting with different SPL,
but based on this example the techniques proposed in this paper look
rather promising.

\begin{table}
\begin{center}
\vspace{10pt}
\begin{tabular}{|l|r|c|r|c|r|c|r|c|} 
\cline{1-9}
\!\!\parbox[t]{3mm}{\multirow{3}{*}{\rotatebox[origin=c]{90}{\sc proper-}}}\parbox[t]{0mm}{\multirow{3}{*}{\rotatebox[origin=c]{90}{\sc ties\ }}}
& \multicolumn{4}{|c|}{\sc product-by-product} & \multicolumn{4}{|c|}{\sc \rule{0pt}{9.5pt}FTS-based family approach} \\
\cline{2-9}
& \multicolumn{2}{|c|}{\!\!\!\sc without bisimulation\!\!\!} & \multicolumn{2}{|c|}{\!\!\!\sc with bisimulation\!\!\!} 
& \multicolumn{2}{|c|}{\!\!\!\sc without bisimulation\!\!\!} & \multicolumn{2}{|c|}{\!\!\!\sc with bisimulation\!\!\!} \\
\cline{2-9}
\!\!& {{\sc time} ({s})\rule{0pt}{8.5pt}} & {\sc result} &
      {{\sc time} ({s})} & {\sc result} & {{\sc time} ({s})} & {\small\sc result} & {{\sc time} ({s})} & {\sc result} \\
\hline
\!\!1\!\!& \rule{0pt}{8.5pt}42.04 & {\sc{false}} & 38.18 & {\sc{true}} & 52.96 & {\sc{false}} & 13.60 & {\sc{true}} \\
\!\!2\!\!& 41.78 & {\sc{true}} & 41.65 & {\sc{true}} & 53.86 & {\sc{true}} & 53.69 & {\sc{true}} \\
\!\!3a\!\!& 42.32 & {\sc{true}} & 37.76 & {\sc{true}} & 70.57 & {\sc{true}} & 7.70 & {\sc{true}} \\
\!\!3b\!\!& 42.01 & {\sc{true}} & 37.78 & {\sc{true}} & 59.96 & {\sc{true}} & 7.98 & {\sc{true}} \\
\!\!4a\!\!& 40.62 & {\sc{true}} & 38.00 & {\sc{true}} & 24.18 & {\sc{true}} & 8.65 & {\sc{true}} \\
\!\!4b\!\!& 40.20 & {\sc{true}} & 37.88 & {\sc{true}} & 20.78 & {\sc{true}} & 10.68 & {\sc{true}} \\
\!\!5a\!\!& 42.38 & {\sc{false}} & 38.51 & {\sc{true}} & 66.08 & {\sc{false}} & 18.59 & {\sc{true}} \\
\!\!5b\!\!& 42.34 & {\sc{false}} & 38.09 & {\sc{true}} & 69.95 & {\sc{false}} & 14.92 & {\sc{true}} \\
\!\!6\!\!& 43.63 & {\sc{true}} & 39.17 & {\sc{true}} & 105.35 & {\sc{true}} & 29.72 & {\sc{true}} \\
\!\!7a\!\!& 42.45 & {\sc{true}} & 38.19 & {\sc{true}} & 71.07 & {\sc{true}} & 13.84 & {\sc{true}} \\
\!\!7b\!\!& 42.35 & {\sc{true}} & 38.04 & {\sc{true}} & 79.05 & {\sc{true}} & 9.48 & {\sc{true}} \\
\!\!8\!\!& 42.82 & {\sc{true}} & 39.09 & {\sc{true}} & 80.69 & {\sc{true}} & 20.47 & {\sc{true}} \\
\hline
\!\!\!{\sc tot}\!\!\!& \rule{0pt}{8.5pt}504.94 & & 462.34 & & 754.50 & & 209.32 & \\
\cline{1-9}
\end{tabular}
\end{center} 
\vspace{0.25cm}
\caption{\label{tab:results}Experimental evaluation results (time in seconds)}
\end{table}




\section{Concluding remarks}
\label{sec-conclusion}

We have defined a novel notion of branching feature bisimilarity for
FTS and an algorithm to minimize an FTS modulo coherent branching
feature bisimulation. This complements and formalizes part of the
feature-oriented modular verification approach of SPL with mCRL2 that
we outlined in~\cite{BV14:formalise,BV14b}. An initial application of
the minimization algorithm to a simplistic SPL promises significant
verification speed-ups.

It remains to establish the subset of the modal $\mu$-calculus that is
preserved by (coherent) branching feature bisimulation, i.e.\ what
properties are respected by our reduction technique. It is
  known that branching bisimulation preserves modal $\mu$-formula
  without the next
  operator~\cite{DV95:jacm}. Theorem~\ref{th-projection} may be used
  to lift the result to branching feature bisimulation, if the
  property $\calS \models \varphi$ iff $\calS_P \models \varphi$ is to
hold. We leave this to future work. It would also be
interesting to see whether the minimization algorithm's complexity can
be reduced, possibly by lifting some optimizations from the Groote \&
Vaandrager algorithm for LTS to our FTS setting, or split multiple
blocks based on a single splitter. 

Finally, we plan to evaluate our modular verification approach on a
more realistic SPL. By expanding the SPL of a coffee vending machine to
   examples growing in size, we may see if the exponential blow-up forecast
  by the NP-completeness result of Theorem~\ref{thm:algo-np-hard} can
  be traced, in particular to observe at what point reduction time
  outweighs the gain of family-based verification. As noted by one of
  the reviewers, family-based verification approaches perform better
  on larger models (both in terms of states and variability), whereas
  reduction techniques are difficult to apply on real, industrial
  models. We hope that the idea, sketched in~\cite{BV14b}, to exploit the
  inherent modular structure of SPL to guide the abstraction, will
  prove fruitful in finding balance in this trade-off and help to
  come up with automated support to reduce a system given a
  property. For this it is useful to reconstruct the experiments
  reported in~\cite{CCPSHL12} and to compare the performance
  gain. Also a study of the relationship of the preorder proposed
  in~\cite{CCPSHL12} to the equivalences put forward here, is an
  interesting topic of research that may increase our understanding of
  the interplay between variability and internal behaviour.

\subsection*{Acknowledgements}
\quarterlineup
Maurice ter Beek was supported by the EU FP7-ICT FET-Proactive project QUANTICOL (600708) and by the Italian MIUR project CINA (PRIN 2010LHT4KM).


\bibliographystyle{eptcs}
\bibliography{fmsple15}

\appendix
\section{Example SPL}
\label{app:SPL}

Here we provide the details of the example SPL used for the
experiments described in Section~\ref{sec-experiments}.  It is an
extension of the coffee vending machine described
in~\cite{ABFG11b,ABFG12,BV14b,BV14:formalise} with a soup component
running in parallel with the usual beverage component. It has the
following list of functional requirements:
%
\begin{itemize}\itemsep=-1pt\parsep=-1pt
\item Each product contains a beverage component. Optionally, 
also a soup component is present.
\item Initially, either a euro must be inserted, exclusively for 
European products, or a dollar must be inserted, exclusively for 
Canadian products. The money can be inserted in either of the 
components.
\item Optionally, money inserted in a component can be retrieved 
via a cancel button, after which money can be inserted in this 
component anew.
\item If money was inserted in the beverage component, the user 
has to choose whether (s)he wants sugar, by pressing one of two 
buttons, after which (s)he can select a beverage.
\item The choice of beverage (coffee, tea, cappuccino) varies, but 
coffee must be offered by all products whereas cappuccino may 
be offered solely by European products.
\item Optionally, a ringtone may be rung after delivering a 
beverage. However, a ringtone must be rung by all products 
offering cappuccino.
\item After the beverage is taken, money can be inserted again in 
the beverage component.
\item If money was inserted in the soup component, the user has 
to choose a type of soup (chicken, tomato, pea). The types of 
soup offered vary, but at least one type must be offered by all 
products with a soup component.
\item The soup component does not contain cups to serve the 
soup in. Hence, the user has to place a cup to pour the soup in. 
Optionally, a cup detector may be present in the soup component. 
It is required that all Canadian products with a soup component
are equipped with a cup detector.
\item If cup detection is present, the chosen type of soup will 
only be delivered after a cup has been detected by the soup 
component. However, the cup detector may fail to detect an 
already placed cup, after which the user will have to place it
again. If a cancel option is available, the user may cancel the
order as long as no cup has been detected. 
\item If cup detection is not present, the soup will be delivered 
immediately after a type of soup was chosen, regardless of 
whether a cup was placed. If no cup was placed there will be 
no soup to take.
\item Optionally, a ringtone (shared with the beverage 
component) may be rung after delivering soup.
\item If a cup was present, money can be inserted again 
in the soup component after the soup is taken.
\end{itemize}
%
\noindent
These yield the attributed feature model in
Figure~\ref{fig:FD} and the behavioral models in
Figures~\ref{fig:drink} and~\ref{fig:soup}.\pagebreak

\newcommand{\mytwolines}[2]{%
  \def\arraystretch{0.75}
  \begin{tabular}{c}{#1}\\{\small #2}\end{tabular}
  \def\arraystretch{1.0}}
\newcommand{\mythreelines}[3]{%
  \def\arraystretch{0.75}
  \begin{tabular}{c}{#1}\\{\small #2}\\{\small #3}\end{tabular}
  \def\arraystretch{1.0}}

\begin{figure}[h]
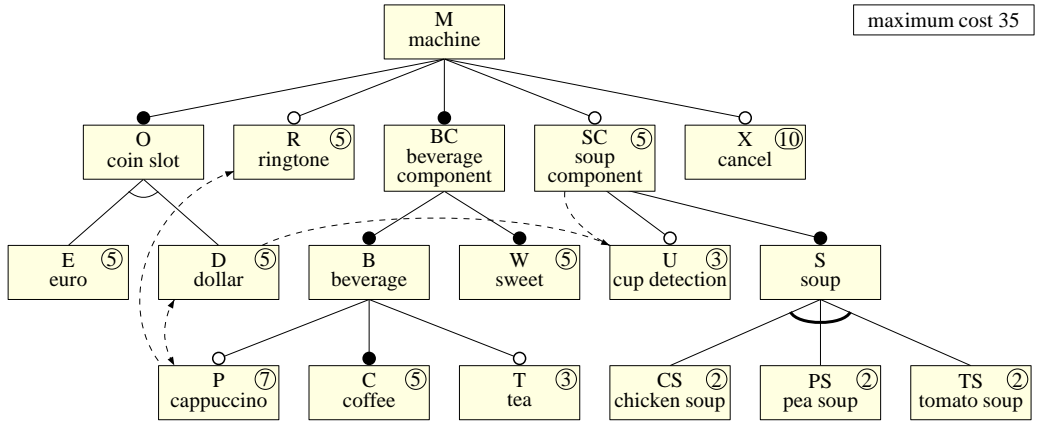

  \hspace*{2.5cm}
  \small
  \scalebox{0.8}{%
  \begin{graph}(100,100)
    \graphset{Nmr=0,Nadjust=n,Nw=20,Nh=9}
    \node[Nframe=n,Nw=0,Nh=0](SCx)(75,65){}
    \node[Nmr=1,Nw=2,Nh=2,linewidth=0.25](optU)(87.5,55.7){}
    \node[Nmr=1,Nw=2,Nh=2,linewidth=0.25,fillcolor=Black](manS)(112.5,55.7){}
    \edge(SCx,optU){}
    \edge(SCx,manS){}
    \node[fillcolor=LightYellow](M)(50,90)
         {\mytwolines{M}{machine}}
    \node[fillcolor=LightYellow](O)(0,70)
         {\mytwolines{O}{coin slot}}
    \node[fillcolor=LightYellow](R)(25,70)
         {\mytwolines{R}{ringtone}}
    \node[fillcolor=LightYellow,Nh=11](BC)(50,69)
         {\mythreelines{BC}{beverage}{component}}
    \node[fillcolor=LightYellow,Nh=11](SC)(75,69)
         {\mythreelines{SC}{soup}{component}}
    \node[fillcolor=LightYellow](X)(100,70)
         {\mytwolines{X}{cancel}}
    \node[fillcolor=LightYellow](E)(-12.5,50)
         {\mytwolines{E}{euro}}
    \node[fillcolor=LightYellow](D)(12.5,50)
         {\mytwolines{D}{dollar}}
    \node[fillcolor=LightYellow](B)(37.5,50)
         {\mytwolines{B}{beverage}}
    \node[fillcolor=LightYellow](W)(62.5,50)
         {\mytwolines{W}{sweet}}
    \node[fillcolor=LightYellow](U)(87.5,50)
         {\mytwolines{U}{cup detection}}
    \node[fillcolor=LightYellow](S)(112.5,50)
         {\mytwolines{S}{soup}}
    \node[fillcolor=LightYellow](P)(12.5,30)
         {\mytwolines{P}{cappuccino}}
    \node[fillcolor=LightYellow](C)(37.5,30)
         {\mytwolines{C}{coffee}}
    \node[fillcolor=LightYellow](T)(62.5,30)
         {\mytwolines{T}{tea}}
    \node[fillcolor=LightYellow](CS)(87.5,30)
         {\mytwolines{CS}{chicken soup}}
    \node[fillcolor=LightYellow](PS)(112.5,30)
         {\mytwolines{PS}{pea soup}}
    \node[fillcolor=LightYellow](TS)(137.5,30)
         {\mytwolines{TS}{tomato soup}}
    \node[Nadjust=n,Nw=4.5,Nh=3.4,Nmr=3.4](Xo)(107.25,72.25){\small 10}
    \node[Nadjust=n,Nw=3.4,Nh=3.4,Nmr=3](SCo)(82.75,72.25){\small 5}
    \node[Nadjust=n,Nw=3.4,Nh=3.4,Nmr=3](Ro)(32.75,72.25){\small 5}
    \node[Nadjust=n,Nw=3.4,Nh=3.4,Nmr=3](Eo)(-5.25,52.25){\small 5}
    \node[Nadjust=n,Nw=3.4,Nh=3.4,Nmr=3](Do)(20.25,52.25){\small 5}
    \node[Nadjust=n,Nw=3.4,Nh=3.4,Nmr=3](Wo)(70.25,52.25){\small 5}
    \node[Nadjust=n,Nw=3.4,Nh=3.4,Nmr=3](Uo)(95.25,52.25){\small 3}
    \node[Nadjust=n,Nw=3.4,Nh=3.4,Nmr=3](Po)(20.25,32.25){\small 7}
    \node[Nadjust=n,Nw=3.4,Nh=3.4,Nmr=3](Co)(45.25,32.25){\small 5}
    \node[Nadjust=n,Nw=3.4,Nh=3.4,Nmr=3](To)(70.25,32.25){\small 3}
    \node[Nadjust=n,Nw=3.4,Nh=3.4,Nmr=3](CSo)(95.25,32.25){\small 2}
    \node[Nadjust=n,Nw=3.4,Nh=3.4,Nmr=3](PSo)(120.25,32.25){\small 2}
    \node[Nadjust=n,Nw=3.4,Nh=3.4,Nmr=3](TSo)(145.25,32.25){\small 2}
    \node[Nh=5,Nw=30](costs)(133,92){\small maximum cost 35}
    \node[Nframe=n,Nw=0,Nh=0](Mx)(50,85.5){}
    \node[Nmr=1,Nw=2,Nh=2,linewidth=0.25,fillcolor=Black](manO)(0,75.7){}
    \node[Nmr=1,Nw=2,Nh=2,linewidth=0.25,fillcolor=Black](manBC)(50,75.7){}
    \node[Nmr=1,Nw=2,Nh=2,linewidth=0.25](optR)(25,75.7){}
    \node[Nmr=1,Nw=2,Nh=2,linewidth=0.25](optSC)(75,75.7){}
    \node[Nmr=1,Nw=2,Nh=2,linewidth=0.25](optX)(100,75.7){}
    \edge(Mx,manO){}
    \edge(Mx,manBC){}
    \edge(Mx,optR){}
    \edge(Mx,optSC){}
    \edge(Mx,optX){}
    \node[Nframe=n,Nw=0,Nh=0](Ox)(0,65.5){}
    \node[Nframe=n,Nw=0,Nh=0](altE)(-12.5,54.5){}
    \node[Nframe=n,Nw=0,Nh=0](altD)(12.5,54.5){}
    \edge(Ox,altE){}
    \edge(Ox,altD){}
    \node[Nframe=n,Nw=0,Nh=0](Oleft)(-2.5,63.5){} 
    \node[Nframe=n,Nw=0,Nh=0](Oright)(2.5,63.5){} 
    \edge[curvedepth=-1](Oleft,Oright){}
    \node[Nframe=n,Nw=0,Nh=0](BCx)(50,63.5){}
    \node[Nmr=1,Nw=2,Nh=2,linewidth=0.25,fillcolor=Black](manW)(37.5,55.7){}
    \node[Nmr=1,Nw=2,Nh=2,linewidth=0.25,fillcolor=Black](manB)(62.5,55.7){}
    \edge(BCx,manW){}
    \edge(BCx,manB){}
    \node[Nframe=n,Nw=0,Nh=0](Bx)(37.5,45.5){}
    \node[Nmr=1,Nw=2,Nh=2,linewidth=0.25](optP)(12.5,35.7){}
    \node[Nmr=1,Nw=2,Nh=2,linewidth=0.25,fillcolor=Black](manC)(37.5,35.7){}
    \node[Nmr=1,Nw=2,Nh=2,linewidth=0.25](optT)(62.5,35.7){}
    \edge(Bx,optP){}
    \edge(Bx,manC){}
    \edge(Bx,optT){}
    \node[Nframe=n,Nw=0,Nh=0](Sx)(112.5,45.5){}
    \node[Nframe=n,Nw=0,Nh=0](altCS)(87.5,34.5){}
    \node[Nframe=n,Nw=0,Nh=0](altPS)(112.5,34.5){}
    \node[Nframe=n,Nw=0,Nh=0](altTS)(137.5,34.5){}
    \edge(Sx,altCS){}
    \edge(Sx,altPS){}
    \edge(Sx,altTS){}
    \node[Nframe=n,Nw=0,Nh=0](Sleft)(107.5,43.5){} 
    \node[Nframe=n,Nw=0,Nh=0](Sright)(117.5,43.5){} 
    \drawbpedge[linewidth=0.5](Sleft,-66,2.65,Sright,-114,2.65){}  
    \drawbpedge[dash={1}0,AHnb=1,ATnb=1,sxo=-2.5,exo=-2.5](D,-150,10,P,150,10){}
    \drawbpedge[dash={1}0,AHnb=1,sxo=+2.5,exo=-5.3](D,+75,12.5,U,+105,12.5){}
    \drawbpedge[dash={1}0,sxo=+2.5,exo=-5.6,eyo=+2.8](SC,-165,15,U,+165,10){}
    \drawbpedge[dash={1}0,ATnb=1,sxo=-2.5,syo=-2,exo=-3.5](R,-180,22.5,P,155,20){}
  \end{graph}
  }
  \vspace*{-2.25cm}
  \caption{\label{fig:FD}Feature model of family of coffee vending machines}
\end{figure}

\noindent
In the attributed feature model, mandatory (core) features are marked by a
closed bullet, optional features by an open one.
Exactly one of the features~$E$ and~$D$ is selected, while at least one of the features $CS$, $PS$ and~$TS$ is selected.
As to cross-tree constraints, 
features $P$ and~$D$ exclude each other, feature~$P$
requires feature~$R$, and the simultaneous selection of features~$D$ and~$SC$ requires feature~$U$. 
The value of the cost attribute of the concrete features is put inside a small circle (i.e.\ $\mathit{cost}(X) = 10$).
Finally, as an additional constraint, we require that the total costs
of all selected features does not exceed the threshold~$35$.

\begin{figure}[h!]
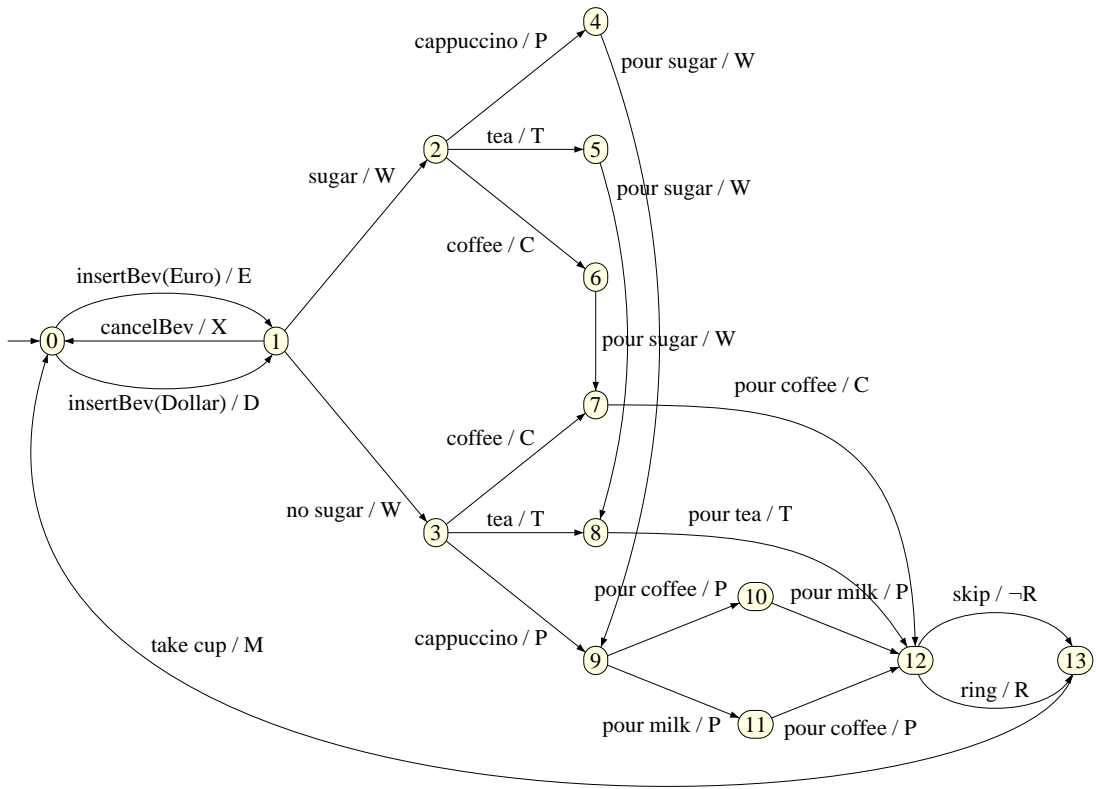

  \centering
  \vspace*{-0.5cm}
  \small
  \scalebox{0.85}{%
    \begin{digraph}(170,125)(-5,25)
      \graphset{iangle=180,fangle=270} 
      		\node[fillcolor=LightYellow](0)(0,90){$0$}
      		\node[fillcolor=LightYellow](1)(35,90){$1$}
      		\node[fillcolor=LightYellow](2)(60,120){$2$}
      		\node[fillcolor=LightYellow](3)(60,60){$3$}
      		\node[fillcolor=LightYellow](4)(85,140){$4$}
      		\node[fillcolor=LightYellow](5)(85,120){$5$}
      		\node[fillcolor=LightYellow](6)(85,100){$6$}
      		\node[fillcolor=LightYellow](7)(85,80){$7$}
      		\node[fillcolor=LightYellow](8)(85,60){$8$}
      		\node[fillcolor=LightYellow](9)(85,40){$9$}
      		\node[fillcolor=LightYellow](10)(110,50){$10$}
      		\node[fillcolor=LightYellow](11)(110,30){$11$}
      		\node[fillcolor=LightYellow](12)(135,40){$12$}
      		\node[fillcolor=LightYellow](13)(160,40){$13$}
      		\imark(0)
      		\drawbpedge[ELside=l](0,90,10,1,90,10){insertBev(Euro) / E}
      		\drawbpedge[ELside=r](0,270,10,1,270,10){insertBev(Dollar) / D}
      		\edge[ELside=r](1,0){cancelBev / X}
      		
      		\edge[ELside=l, ELpos=70](1,2){sugar / W}
      		\edge[ELside=l, ELpos=50](2,4){cappuccino / P}
      		\edge[ELside=l, ELpos=50](2,5){tea / T}
      		\edge[ELside=r, ELpos=50](2,6){coffee / C}
      		
      		\edge[ELside=l, ELpos=10, curvedepth=10](4,9){pour sugar / W}
      		\edge[ELside=l, ELpos=15, curvedepth=5](5,8){pour sugar / W}
      		\edge[ELside=l, ELpos=50](6,7){pour sugar / W}
      		
      		\edge[ELside=r, ELpos=70](1,3){no sugar / W}
      		\edge[ELside=r, ELpos=50](3,9){cappuccino / P}
      		\edge[ELside=l, ELpos=50](3,8){tea / T}
      		\edge[ELside=l, ELpos=50](3,7){coffee / C}
      		
      		\drawbpedge[ELside=l, ELpos=40](7,0,30,12,90,40){pour coffee / C}
      		\drawbpedge[ELside=l, ELpos=30](8,0,30,12,120,20){pour tea / T}
      		\edge[ELside=l](9,10){pour coffee / P}
      		\edge[ELside=l](10,12){pour milk / P}
      		\edge[ELside=r](11,12){pour coffee / P}
      		\edge[ELside=r](9,11){pour milk / P}
      		
      		\drawbpedge[ELside=l](12,90,10,13,90,10){skip / $\neg$R}
      		\drawbpedge[ELside=l](12,270,10,13,270,10){ring / R}
      		
      		\drawbpedge[ELside=r, ELpos=70](13,270,30,0,250,90){take cup / M}
    \end{digraph}
}
\vspace*{0.25cm}
\caption{\label{fig:drink}FTS of beverage component}
\end{figure}
 
\newpage
\noindent 
The FTS of the beverage component contains 
$14$~states and $23$~transitions and that of the 
soup component contains $13$~states and~$28$ transitions, 
for a total of~$182$ states and $691$~transitions in parallel 
composition.

\begin{figure}[hb!]
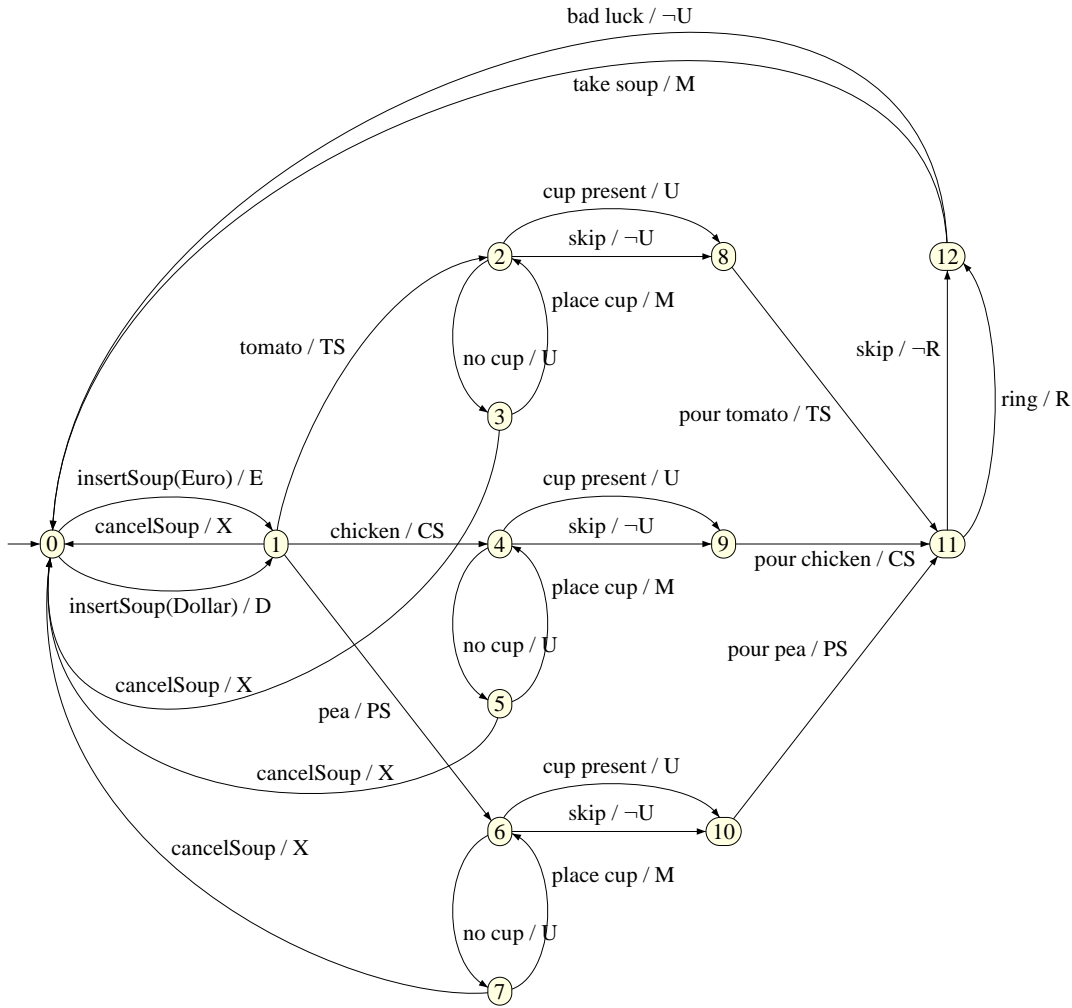

  \centering
  \vspace*{2.5cm}
  \small
  \scalebox{0.85}{%
    \begin{digraph}(170,120)(-5,25)
      \graphset{iangle=180,fangle=270} 
      \node[fillcolor=LightYellow](0)(0,90){$0$}
      \node[fillcolor=LightYellow](1)(35,90){$1$}
      \node[fillcolor=LightYellow](2)(70,135){$2$}
      \node[fillcolor=LightYellow](3)(70,110){$3$}
      \node[fillcolor=LightYellow](4)(70,90){$4$}
      \node[fillcolor=LightYellow](5)(70,65){$5$}
      \node[fillcolor=LightYellow](6)(70,45){$6$}
      \node[fillcolor=LightYellow](7)(70,20){$7$}
      \node[fillcolor=LightYellow](8)(105,135){$8$}
      \node[fillcolor=LightYellow](9)(105,90){$9$}
      \node[fillcolor=LightYellow](10)(105,45){$10$}
      \node[fillcolor=LightYellow](11)(140,90){$11$}
      \node[fillcolor=LightYellow](12)(140,135){$12$}
      \imark(0)
      \drawbpedge[ELside=l](0,75,10,1,90,10){insertSoup(Euro) / E}
      \drawbpedge[ELside=r](0,285,10,1,270,10){insertSoup(Dollar) / D}
      \edge[ELside=r](1,0){cancelSoup / X}
      
      \drawbpedge[ELside=l](1,90,10,2,180,20){tomato / TS}
      \edge[ELside=l, ELpos=50](1,4){chicken / CS}
      \edge[ELside=r, ELpos=50](1,6){pea / PS}
      
      \drawbpedge[ELside=l](2,90,10,8,90,10){cup present / U}
      \edge(2,8){skip / $\neg$U}
      \drawbpedge[ELside=l, ELpos=70](2,190,10,3,170,10){no cup / U}
      \drawbpedge[ELside=r, ELpos=60](3,0,10,2,0,10){place cup / M}
      \drawbpedge[ELside=r, ELpos=60](3,270,40,0,260,50){cancelSoup / X}
      \edge[ELside=r, ELpos=40](8,11){pour tomato / TS}
      
      \drawbpedge[ELside=l](4,90,10,9,90,10){cup present / U}
      \edge(4,9){skip / $\neg$U}
      \drawbpedge[ELside=l, ELpos=70](4,190,10,5,170,10){no cup / U}
      \drawbpedge[ELside=r, ELpos=60](5,0,10,4,0,10){place cup / M}
      \drawbpedge[ELside=r, ELpos=40](5,270,20,0,260,50){cancelSoup / X}
      \edge[ELside=r, ELpos=50](9,11){pour chicken / CS}
      
      \drawbpedge[ELside=l](6,90,10,10,90,10){cup present / U}
      \edge(6,10){skip / $\neg$U}
      \drawbpedge[ELside=l, ELpos=70](6,190,10,7,170,10){no cup / U}
      \drawbpedge[ELside=r, ELpos=60](7,0,10,6,0,10){place cup / M}
      \drawbpedge[ELside=r, ELpos=50](7,190,20,0,260,50){cancelSoup / X}
      \edge[ELside=l, ELpos=50](10,11){pour pea / PS}
      		
      \drawbpedge[ELside=l,ELpos=65](11,90,10,12,270,10){skip / $\neg$R}
      \drawbpedge[ELside=r](11,0,10,12,0,10){ring / R}
      \drawbpedge[ELside=l,ELpos=40](12,90,60,0,90,70){take soup / M}
      \drawbpedge[ELside=r,ELpos=40](12,90,70,0,90,70){bad luck / $\neg$U}
    \end{digraph}
}
\vspace*{0.5cm}
\caption{\label{fig:soup}FTS of soup component}
\end{figure}

As reported in Section~\ref{sec-experiments}, we used the 
mCRL2 toolset to verify~$12$ properties against this SPL\@. 
These properties are listed next, together with their 
formalization in the mCRL2 variant of the modal $\mu$-calculus.
%
\begin{enumerate}\itemsep=-1pt\parsep=-1pt
\item[1.] If a coffee is ordered, then eventually coffee is poured:
$[\,\textit{true}*.\,\textit{coffee}\,]\,(\text{mu}\,X.\,[\,!\,\textit{pour\_coffee}\,]\,X)$
\item[2.] The SPL is deadlock-free: 
$[\,\textit{true}*\,]\,\langle\textit{true}\rangle\,\textit{true}$
\item[3a.] A machine that accepts Euros does not accept Dollars:\\
$[\textit{true}*.(\textit{insertBev}(\textit{Euro})\mid\mid\textit{insertSoup}(\textit{Euro})).\textit{true}*.(\textit{insertBev}(\textit{Dollar})\mid\mid\textit{insertSoup}(\textit{Dollar}))]\,\textit{false}$
\item[3b.] A machine that accepts Dollars does not accept Euros:\\
$[\textit{true}*.(\textit{insertBev}(\textit{Dollar})\mid\mid\textit{insertSoup}(\textit{Dollar})).\textit{true}*.(\textit{insertBev}(\textit{Euro})\mid\mid\textit{insertSoup}(\textit{Euro}))]\,\textit{false}$
\item[4a.] A cup can only be taken out of the beverage component after a beverage was ordered:\\ 
$[\,(!\,\textit{coffee}\ \&\&\ !\,\textit{tea}\ \&\&\ !\,\textit{cappuccino})*.\,\textit{take\_cup}\,]\ \textit{false}$
\item[4b.] A cup can only be taken out of the soup component after soup was ordered:\\
$[\,(!\,\textit{tomato}\ \&\&\ !\,\textit{chicken}\ \&\&\ !\,\textit{pea})*.\,\textit{take\_soup}]\ \textit{false}$
\item[5a.] If a beverage is ordered, then eventually the beverage is canceled or a cup is taken out of the beverage component: 
$[\,\textit{true}*.\,(\,\textit{coffee}\,\mid\mid\,\textit{tea}\,\mid\mid\,\textit{cappuccino})\,]\,(\text{mu}\,X.\,[\,(!\,\textit{cancelBev}\ \&\&\ !\,\textit{take\_cup})\,]\,X)$
\item[5b.] If soup is ordered, then eventually the soup is canceled, a cup is taken out of the soup component or the customer has bad luck:\\ 
$[\,\textit{true}*.\,(\textit{tomato}\,\mid\mid\,\textit{chicken}\,\mid\mid\,\textit{pea})\,]\,(\text{mu}\,X.\,[\,(!\,\textit{cancelSoup}\ \&\&\ !\,\textit{take\_soup}\ \&\&\ !\,\textit{bad\_luck})\,]\,X)$
\item[6.] If the machine has a soup component, then a beverage can be ordered without inserting more money after soup was ordered: 
$[\,\textit{true}*.\,(\textit{insertSoup}(\textit{Euro})\mid\mid\textit{insertSoup}(\textit{Dollar}))\,]\,\langle\textit{true}*.\,(\textit{tomato}\,\mid\mid \textit{chicken}\,\mid\mid\,\textit{pea}).\,(!\,\textit{insertBev}(\textit{Euro})\ \&\&\ !\,\textit{insertBev}(\textit{Dollar}))*.\,(\textit{coffee}\,\mid\mid\,\textit{tea}\,\mid\mid\,\textit{cappuccino})\rangle\,\textit{true}$
\item[7a.] A beverage cannot be ordered without inserting more money if a previous beverage order is still pending: 
$[\textit{true}*\!.(\textit{coffee}\mid\mid\textit{tea}\mid\mid\textit{cappuccino}).(!\textit{insertBev}(\textit{Dollar})\ \&\&\,!\textit{insertBev}(\textit{Euro}))*.(\textit{coffee}\mid\mid\textit{tea}\mid\mid\textit{cappuccino})]\ \textit{false}$
\item[7b.] Soup cannot be ordered without inserting more money if a soup order is pending: $[\,\textit{true}*.\,(\textit{tomato}\,\mid\mid 
\textit{chicken}\,\mid\mid\,\textit{pea}).\ (!\,\textit{insertSoup}(\textit{Dollar})\ \&\&\ !\,\textit{insertSoup}(\textit{Euro}))*.\ (\textit{tomato}\,\mid\mid\,\textit{chicken}\,\mid\mid\,\textit{pea})\,]\ \,\textit{false}$
\item[8.] In\,a\,machine\,with\,cup\,detection,\,soup\,can\,only\,be\,poured\,after\,detecting\,a\,cup:\,$[\,\textit{true}*.\,\textit{cup\_present}\,]$\\
$[\textit{true}*.\,(\textit{take\_soup}\mid\mid\textit{bad\_luck}).\,(!\,\textit{cup\_present})*.\,(\textit{pour\_tomato}\mid\mid\textit{pour\_chicken}\mid\mid\textit{pour\_pea})]\ \textit{false}$
\end{enumerate}

\vspace*{2.75cm}

\end{document}